\documentclass{article}
\usepackage{graphicx}
\usepackage{epsfig}
\usepackage{amssymb}
\usepackage{amsfonts}
\usepackage{amsmath}
\usepackage{amsthm}
\newtheorem{thm}{Theorem}[section]
\newtheorem{dfn}[thm]{Definition}
\newtheorem{lem}[thm]{Lemma}
\newtheorem{rem}[thm]{Remark}
\newtheorem{cor}[thm]{Corollary}

\newtheorem{asm}[thm]{Assumption}
\newtheorem{exm}[thm]{Example}
\numberwithin{equation}{section}

\numberwithin{equation}{section}

\begin{document}
\title{Hedging of Game Options With the Presence of Transaction Costs}
\date{9.3.2012}
\author{
  Yan Dolinsky
  \thanks{
  ETH Zurich, Dept. of Mathematics,
  \texttt{yan.dolinsky@math.ethz.ch}
  }}

\maketitle

\maketitle
\begin{abstract}
We study the problem of super--replication for game
options under proportional transaction costs.
We consider a multidimensional continuous time model,
in which the discounted stock price process satisfies
the conditional full support property.
We show that the super--replication price is the cheapest cost of a trivial
super--replication strategy. This result is an extension of
previous papers (see \cite{B} and \cite{GRS})
which considered only European options.
In these papers the authors showed that with the presence of
proportional transaction costs the super--replication price of a European option
is given in terms of the concave envelope of the payoff function.
In the present work we prove that for game options the super--replication price
is given by a game
variant analog of the standard concave envelope term. The treatment of game
options is more complicated and requires additional tools. We combine
the theory of consistent price systems together with
the theory of extended weak convergence
which was developed in \cite{A}.
The second theory is essential in dealing with hedging which involves
stopping times, like in the case of game options.
\end{abstract}
{\small


\section{Introduction}\label{sec:1}\setcounter{equation}{0}
This paper deals with the super--replication of cash--settled game (Israeli) options
in the presence of proportional transaction costs.
A game contingent claim (GCC) or game option which was introduced in
\cite{Ki1}
is defined
as a contract between the seller and the buyer of the option such
that both have the right to exercise it at any time up to a
maturity date (horizon) $T$. If the buyer exercises the contract
at time $t$ then he receives the payment $Y(t)$, but if the seller
exercises (cancels) the contract before the buyer then the latter
receives $X(t)$. The difference $\Delta(t)=X(t)-Y(t)$ is the penalty
which the seller pays to the buyer for the contract cancellation.
In short, if the seller will exercise at a stopping time
$\sigma\leq{T}$ and the buyer at a stopping time $\tau\leq{T}$
then the former pays to the latter the amount $H(\sigma,\tau)$
where
\begin{equation*}
H(\sigma,\tau)=X(\sigma)\mathbb{I}_{\sigma<\tau}+Y(\tau)\mathbb{I}_{\tau\leq{\sigma}}
\end{equation*}
and we set $\mathbb{I}_{Q}=1$ if an event $Q$ occurs and
$\mathbb{I}_{Q}=0$ if not.

A hedge (for the seller) against a GCC is defined as a pair
$(\pi,\sigma)$ which consists of a self financing strategy $\pi$
and a stopping time $\sigma$ which is the cancellation time
for the seller. A hedge is called perfect if no matter what exercise
time the buyer chooses, the seller can cover his liability to the
buyer (with probability one). Since our contingent claim is
cash--settled, we measure the portfolio value in cash, assuming that there
are no liquidity costs for turning stocks into cash
in the exercise moment of the options.
The option price $V^*$ is defined as
the minimal initial capital which is required for a perfect hedge,
i.e. for any $\Xi>V^*$ there is a perfect hedge with an initial
capital $\Xi$.

We consider a general model of financial market which consists of a
savings account with a stochastic interest rate
and $d$ stocks which are given by a continuous stochastic process.
We assume that the discounted stock price process satisfies the conditional full support property
which was introduced in \cite{GRS}. In general,
the conditional full support property is quite general assumption.
In particular, processes such as Markov diffusions,
solutions of SDEs in
the Brownian setup with path dependent coefficient
(under some regularity conditions)
and fractional Brownian motion
satisfy this assumption (for details see \cite{GRS} and \cite{P}).

Our main result states that the super--replication price is the cheapest cost of a
trivial perfect hedge.
For game options a trivial hedge is a pair which consists of a buy--and--hold
strategy and a hitting time of the stock process into a Borel set. Furthermore, we find explicit
formulas for the cheapest perfect hedge, and characterize the super--replication value
as the game analog of the standard concave envelope which appears in the European options case.
We provide several examples for explicit calculations of the super--replication prices
together with the optimal hedges.

These results are an extension of previous results which were obtained
for European options,
see for example, \cite{B}, \cite{BT}, \cite{GRS},  \cite{LS} and
\cite{SSC}. The most general results were proved in
\cite{B} and \cite{GRS}
where the authors only assumed the conditional
full support property of the (discounted) stock process.
In all of the above papers the authors showed that the
super--replication
price is given in terms of the concave
envelope of the payoff function, and the way to achieve
this price is by using buy--and--hold strategies.

Our main tool is the consistent price systems
approach which was proven to be very powerful for
European options (see \cite{B}, \cite{GRS}).
We derive a family of consistent price systems which
converge weakly to Brownian martingales of general type.
This together with the theory of extended weak convergence
allows us to bound from below the super--replication price by the value
of some robust optimization problem on the Brownian probability space.
The value of this robust optimization problem
leads to the notion of game variant of the concave envelope.
This notion is also appears naturally in the static
super--replication of game options.

The paper is organized as follows. Main results of this paper are
formulated in the next section, where we also give few examples
of applications of these results.
In Section 3 we derive a general family of consistent price systems.
In Section 4 we treat a robust optimization problem and
establish a connection between the value of this problem
and the game analog of the concave envelope.
Furthermore we use a convex analysis to show that the latter
concept characterizes the static super--hedging price.
In Section 5 we use the extended weak convergence theory in order to
prove an essential limit theorem which evolve optimal stopping
and consistent price systems.
In Section 6 we
complete the proof of Theorem \ref{thm2.1} which is the main result of the paper.

\section{Preliminaries and main results}\label{sec:2}\setcounter{equation}{0}
Let $(\Omega,\mathcal{F},P)$ be a complete probability space together with a filtration
$\{\mathcal{F}_t\}_{t=0}^T$ which satisfies the usual conditions
where $T<\infty$ is a fixed maturity date.
Our financial market consists of a bond (savings account) $S_0(t)$
and of $d$ stocks given by a continuous adapted process
$S:=\{S_1(t),...,S_d(t)\}_{t=0}^T$ which takes on values in $\mathbb{R}^d_{++}$.
We will assume that the bond price is of the form
\begin{equation*}\label{2.8}
S_0(t)=\exp\left(\int_{0}^t r(u)du\right)
\end{equation*}
where $\{{r(t)\}}_{t=0}^T$ is a non--negative adapted
process which represents the interest rate of the savings account. Without loss of generality we assume
that $S_0(0)=1$.
As usual
when we deal with hedging it is convenient to work with the discounted terms.
Thus, we introduce the discounted stock price
\begin{equation*}\label{2.9}
\tilde{S}_i(t)=\frac{S_i(t)}{S_0(t)} \ \ \ 1\leq i\leq d, \ \ t\in [0,T].
\end{equation*}
Before introducing the assumption of conditional full support, we review some
concepts. For any $t<T$
consider the space ${C}^{+}([t,T];\mathbb{R}^d)$
of all continuous functions $f:[t,T]\rightarrow\mathbb{R}^d_{++}$
endowed with the uniform topology. As usual, the support of a a probability measure $\mathbb{P}$
on a separable space is
denoted by $supp$ $\mathbb{P}$ and it is defined as the minimal closed set of measure $1$.
We will also use the notation
${C}^{+}_z([t,T];\mathbb{R}^d)$
for the space of all functions $f\in {C}^{+}([t,T];\mathbb{R}^d)$
which start at $z$, namely $f(t)=z$.
\begin{asm}\label{asm2.1}
The process $\tilde{S}$ is satisfies the conditional full support property with respect to the filtration
${\{\mathcal{F}_t\}}_{t=0}^T$.
Namely, for all $t\in [0,T)$
\begin{equation*}\label{2.9+}
 supp \ P(\tilde{S}_{|[t,T]}|\mathcal{F}_t)=C^{+}_{\tilde{S}(t)}([t,T];\mathbb{R}^d) \ \ \mbox{a.s.}
\end{equation*}
where  $P(\tilde{S}_{|[t,T]}|\mathcal{F}_t)$
denotes the $\mathcal{F}_t$--conditional
distribution of the ${C}^{+}([t,T];$
$\mathbb{R}^d)$--valued
random variable $\tilde{S}_{|[t,T]}$.
\end{asm}

Again, let us emphasize that
Markov diffusions,
solutions of SDEs in
the Brownian setup with path dependent coefficient
(under some regularity conditions)
and fractional Brownian motion
satisfy the above assumption (for deteails see \cite{GRS} and \cite{P}).

We also assume that the interest rate process is bounded uniformly by some constant $H$, i.e.
$r\leq H$, $P\otimes\lambda$ a.s, where $\lambda$ is the Lebesgue measure
on $[0,T]$. In Example \ref{exm2.4} we show
that without this assumption our
main results (which are formulated in Theorem \ref{thm2.1})
should not hold true.

Let $F:\mathbb{R}^d_{+}\rightarrow\mathbb{R}_{+}$
be a convex Lipschitz continuous function and let $\Delta>0$ be a constant.
Consider a game option with the discounted payoff
processes
\begin{equation*}\label{2.10}
Y(t)=\frac{1}{S_0(t)}F(S(t)) \ \ \mbox{and}\ \
X(t)=\frac{1}{S_0(t)}(F(S(t))+\Delta),
\ \ t\in [0,T].
\end{equation*}
Set
\begin{equation*}\label{2.11}
H(t,s)=X(t)\mathbb{I}_{t<s}+Y(s)\mathbb{I}_{s\leq t}, \ \ t,s\in [0,T].
\end{equation*}
Observe that $H(\sigma,\tau)$ is the discounted reward
that the buyer receives given that his exercise time is $\tau$ and the seller
cancellation time is $\sigma$.
Namely we consider game options with non path dependent payoffs and with constant penalty
for the seller's exercise. In general,
for the case where the penalty is non constant,
our results (which are formulated in Theorem \ref{thm2.1})
should not hold true. In particular, even the static
super--replication price may depend on the interest rate process.
This is illustrated in Example \ref{exm2.3}.

Next, let $\kappa\in (0,1)$ be a constant.
We assume that an investor must purchase risky assets through his savings account, i.e.
bartering between two risky assets is impossible.
Consider a model in which for any $1\leq i\leq d$,
every purchase or sale of the $i$--th risky asset at moment $t\in [0,T]$ is subject to a proportional
transaction cost of rate $\kappa$. A trading strategy with an initial capital $\Xi$ is a pair $\pi=(\Xi,\gamma)$
where $\gamma:=\{\gamma_i\}_{1\leq i\leq d}$
such that for any $i$,  $\gamma_i=\{\gamma_i(t)\}_{t=0}^T$ is
an adapted process of bounded variation with left continuous paths.
The random variable $\gamma_i(t)$
denotes the number of shares of the $i$--th asset in the portfolio $\pi$
at moment $t$ (before a transfer is made at this time).
This is exactly the reason why we assume that the process $\gamma$ is left continuous.
The discounted portfolio value of a trading strategy $\pi$ is given by
\begin{eqnarray*}\label{2.13}
&V^{\pi}_{\kappa}(t)=\Xi+\langle \gamma(t),\tilde{S}(t)\rangle-\langle \gamma(0),s\rangle+\\
&(1-\kappa) \int_{[0,t]}\langle \tilde{S}(u),d\gamma_{-}(u)\rangle-
(1+k)\int_{[0,t]}\langle \tilde{S}(u),d\gamma_{+}(u)\rangle,
 \ \ t\in [0,T]\nonumber
\end{eqnarray*}
where  $\langle\cdot,\cdot\rangle$ denotes
the standard scalar product of $\mathbb{R}^d$ and all the
integrals in the above formula are Stieltjes integrals.
As usual $\gamma_{+}(t)=(\gamma_{+,1}(t),...,\gamma_{+,d}(t))$
and $\gamma_{-}(t)=(\gamma_{-,1}(t),...,\gamma_{-,d}(t))$,
where
$\gamma_i(t)=\gamma_{+,i}(t)-\gamma_{-,i}(t)$, $i=1,...,d$ is the Jordan decomposition
into a positive variation $\gamma_{+,i}$ and a negative variation $\gamma_{-,i}$.
Observe that we do not assume any semi--martingale structure of the risky assets.
The term $V^\pi_{\kappa}(t)$ is the (discounted) portfolio
value at time $t$, before a transfer is made at this time. Indeed,
$$\Xi-\langle \gamma(0),s\rangle+(1-\kappa) \int_{[0,t]}\langle \tilde{S}(u),d\gamma_{-}(u)
\rangle-(1+k)\int_{[0,t]}\langle \tilde{S}(u),d\gamma_{+}(u)\rangle$$
is the discounted value of the wealth which is held in the savings account, and
$\langle \gamma(t),\tilde{S}(t)\rangle$ is the discounted value of the wealth which is held in stocks.
The set of all self financing strategies with an initial
capital $\Xi$ will be denoted by $\mathcal{A}(\Xi)$.
Let $\mathcal{T}_{[0,T]}$ be the set of all stopping
times which take on values in $[0,T]$.
A pair $(\pi,\sigma)\in\mathcal{A}(\Xi)
\times \mathcal{T}_{[0,T]}$ of a self financing
strategy $\pi=(\Xi,\gamma)$ and a stopping time
$\sigma$ will be called a hedge. A hedge $(\pi,\sigma)$ will be called trivial if it is of the form
\begin{equation*}\label{2.13+}
\gamma\equiv\gamma(0), \ \ \mbox{and} \ \ \sigma=\inf\{t|S(t)\in D\}\wedge{T}
\end{equation*}
where $D\subset\mathbb{R}^d$ is a Borel set.
Namely we do not trade, and
cancel the option at the first time that the stock process vector
enters a Borel set.
A hedge $(\pi,\sigma)$ will be called perfect if
for any $t\in [0,T]$,
$V^\pi_{\kappa}(t)\geq H(\sigma,t)$ a.s.
It is well known (see Theorem 12.16 in \cite{PM})
that a Stieltjes integral of a continuous function
with respect to a left continuous function of bounded variation,
is also left continuous. Thus
the portfolio value process ${\{V^\pi_{\kappa}(t)\}}_{t=0}^T$ is left continuous and so,
a hedge $(\pi,\sigma)$ is perfect iff
$$P\left(\forall{t}\in [0,T], V^\pi_{\kappa}(t)\geq H(\sigma,t)\right).$$
The super--hedging price is defined by
\begin{equation}\label{2.15}
V_{\kappa}(s)=\inf\{\Xi|\exists(\pi,\sigma)\in\mathcal{A}(\Xi)\times \mathcal{T}_{[0,T]} \ \mbox{which} \ \mbox{is}
\ \mbox{a} \ \mbox{perfect} \ \mbox{hedge}\}
\end{equation}
where $s=S(0)$ is the initial stock position.
We set
\begin{equation}\label{2.15+}
\hat{V}(s)=\inf\{\Xi|\exists(\pi,\sigma)\in\mathcal{A}(\Xi)\times \mathcal{T}_{[0,T]} \ \mbox{which} \ \mbox{is}
\ \mbox{a} \ \mbox{perfect} \ \mbox{and} \ \mbox{a} \ \mbox{trivial} \ \mbox{hedge}\}.
\end{equation}
Since for trivial hedges there are no transaction costs,
$\hat{V}(s)$
does not depend on $\kappa$.
Clearly, $V_{\kappa}(s)\leq \hat{V}(s)$ for any $\kappa$. Notice also that from the
Lipschitz property of $F$ we have $\hat{V}<\infty$.

Before we formulate the main result of the paper,
 we will need some preparations.
Let $\mathcal{G}$ be the set of all functions
$f:\mathbb{R}^d_{+}\rightarrow\mathbb{R}_{+}$ which satisfy the following conditions.\\
\\
i. The function $f$ is continuous and for any $x\in\mathbb{R}^d_{+}$, $F(x)\leq f(x)\leq F(x)+\Delta$.\\
ii. Let $D\subset\mathbb{R}^d_{+}$ be a convex set in which $f(x)<F(x)+\Delta$.
Then $f$ is concave in $D$.

Clearly the function $F+\Delta\in\mathcal{G}$ and so
$\mathcal{G}$ is a non empty set.
It turns out (the proof will be given in Lemma \ref{lem2+.3})
that $\mathcal{G}$ has a minimal element $R\in\mathcal{G}$ which can be calculated explicitly, i.e. $R(x)\leq g(x)$
for any $g\in\mathcal{G}$ and $x\in\mathbb{R}^d_{+}$.
The function $R$ is the game variant of the standard concave envelope. Notice that if $\Delta=\infty$
then $R$ equals to the concave envelop of $F$. The function $R$ can be calculated as following.
For any $1\leq i\leq d$, define
\begin{equation*}\label{2.15++}
F_i(x)=F(0,...,0,x,0,....,0), \ \ x\in\mathbb{R}_{+}
\end{equation*}
where the $x$ appears in the $i$--coordinate above.
Introduce the terms
\begin{equation}\label{2.20}
A_i=\inf\bigg\{t>0\bigg{|}\frac{F_i(t)+\Delta-F(0)}{t}\in\partial{F}_i(t)\bigg\} \ \ \mbox{and}
\ \ B_i=\frac{F_i(A_i)+\Delta-F(0)}{A_i}
\end{equation}
where $\partial{F}_i(t)$ is the sub--gradient of the convex function $F_i$ at $t$.
If $A_i=\infty$, i.e. the set in (\ref{2.20}) is empty then
$B_i=\sup\{v\in\bigcup_{t>0}\partial{F}_i(t)\}<\infty$, (recall that $F$ is Lipschitz continuous).
Observe that for the case $A_i<\infty$, the linear function $F(0)+B_ix$ is the (unique)
tangent from the point $(0,F(0))$ to the function $F_i(x)+\Delta$.
Set $B=(B_1,...,B_d)$ and
define the function $R:\mathbb{R}^d_{+}\rightarrow\mathbb{R}_{+}$
by
\begin{eqnarray}\label{2.20+}
&R(0)=F(0) \ \ \mbox{and} \ \
R(x)=\big(F(0)+\langle B,x\rangle\big)\mathbb{I}_{||x||<H(x)}\\
&+\big(F(x)+\Delta\big)\mathbb{I}_{||x||\geq H(x)}, \ \ \mbox{for} \ \ x\neq 0,\nonumber
\end{eqnarray}
where $H(x)=\inf\{t|F(0)+\langle B,t x/{||x||}\rangle\geq\Delta+F(t x/||x||)\}$
and $H(x)=\infty$ if the above set is empty. Observe that $R\leq F+\Delta$.

The following theorem is the main result of the paper
and it says that the super--replication price is the cheapest cost of a trivial
super--replication strategy, which is equal to $R(s)$, the game variant of the concave envelope.
\begin{thm}\label{thm2.1}
For any  $\kappa\in (0,1)$ and $s\in\mathbb{R}^d_{+}$,
\begin{equation}\label{2.20++}
V_{\kappa}(s)=\hat{V}(s)=R(s).
\end{equation}
Furthermore, let $s\in\mathbb{R}^d_{+}$ be an initial stock position. Define
a trivial hedge  $(\pi,\sigma)$ according to the following cases:\\
i. If $R(s)<F(s)+\Delta$,
\begin{eqnarray}\label{2.21}
&\pi=(R(s),\gamma) \ \ \mbox{where} \ \ \gamma\equiv B \  \ t>0,\\
&\mbox{and} \ \ \sigma=\inf\{t| \Delta+F(S(t))\leq F(0)+\langle B,S(t)\rangle\}\wedge {T}.\nonumber
\end{eqnarray}
ii. If $R(s)=F(s)+\Delta$,
\begin{equation}\label{2.22}
\pi=(R(s),\gamma) \ \ \mbox{where} \ \ \gamma\equiv 0, \ \ \mbox{and} \ \ \sigma=0.
\end{equation}

Then $(\pi,\sigma)$ is a perfect hedge with the smallest initial capital.
\end{thm}
The second case in the above theorem corresponds to a
situation where the initial capital $R(s)$ is equal to
the high payoff $F(s)+\Delta$ and so, the
seller can cancel the contract at the initial moment of time $t=0$
and no actions are needed.
\begin{rem}\label{rem2.1+}
We assume that at the initial moment of time the
investor allowed to have holdings in stocks. Namely,
$\gamma(0)$ is not necessary equal to $0$.
Furthermore, when we calculate the portfolio value at some $t$,
we do not take into account the
liquidation price of the stocks into cash. The reason for this is that although
our options are cash settled, in
real market conditions the stocks can be
delivered physically from the seller to the buyer,
for example, for a Call option's the seller
can give the stock without liquidating it. In the papers \cite{B} and \cite{GRS}
the authors assume that the investor starts with zero stock holdings and must liquidate
his portfolio at the maturity date (the papers deal with European options).
Thus in their setup even trivial strategies are subject to transaction
costs, that is why
the main results in these papers deal only with the asymptotic behaviour (as the rate of the transaction costs
goes to $0$) of the super--replication prices.
\end{rem}
\begin{rem}\label{rem2.2+}
Consider a model with proportional transaction costs of the following
type. The investor is allowed to transfer from the $i$--th
asset to the $j$--th asset for any $0\leq i,j\leq d$,
where the $0$--asset denotes the savings account.
In time $t\in [0,T]$ the above kind of
transfer is subject to proportional transaction costs with a random coefficient
$\lambda^{ij}(t)$. We still allow to the investor
to hold stocks at the initial moment of time $t=0$.
If there exists $\epsilon>0$ such that
$P\big(\min_{0\leq i,j\leq d}\inf_{0\leq t\leq T}\lambda^{ij}(t)>\epsilon\big)=1$,
then there exists $\kappa'\in (0,1)$ such that
$\frac{1-\kappa'}{1+\kappa'}>\frac{1}{1+\lambda^{ij}(t)}$
for any $t\in [0, T]$ and $0\leq i,j\leq d$. Thus
the super--replication price is no less than
$V_{\kappa'}(s)$, and so from Theorem \ref{thm2.1}
we get that the super--replication price in this general setup
is again the cheapest cost of a trivial
super--replication strategy.
\end{rem}
Next, we give three examples for applications of Theorem \ref{thm2.1}.
\begin{exm}[Call option]\label{exm2.1}
Let $K>0$ be a constant and $d=1$ (we have one risky asset which is denoted by $S$).
Consider a game call option
with the discounted payoffs
\begin{equation*}\label{2.23}
Y(t)=\frac{(S(t)-K)^{+}}{S_0(t)} \ \ \mbox{and} \ \ X(t)=Y(t)+\frac{\Delta}{S_0(t)}, \ \ t\in [0,T].
\end{equation*}
Namely, $F(x)=(x-K)^{+}$.
We need to split the analysis into two different cases.\\
i. $\Delta>K$. In this case we have
$A=\infty$ and $B=1$ (recall formulas (\ref{2.20})--(\ref{2.20+})), and so
$R(x)=x$.
From (\ref{2.21}) we get that for any initial stock position $s\in\mathbb{R}_{+}$
the cheapest perfect hedge $(\pi,\sigma)$
is given by
$\pi=(s,\gamma)$ where $\gamma\equiv 1$ and $\sigma=T$.\\
ii. $\Delta\leq K$. In this case we have
$A=K$ and $B=\frac{\Delta}{K}$. Thus (see Fig 1)
\begin{equation*}\label{2.24+}
R(x)=\frac{\Delta x}{K}  \mathbb{I}_{x<K}+(x+\Delta-K) \mathbb{I}_{x\geq K}.
\end{equation*}
Let $s\in\mathbb{R}_{+}$ be an initial position of the stock.
From (\ref{2.21}) we obtain that if $s<K$
then the optimal perfect hedge is given by
\begin{equation*}\label{2.24++}
\pi=\left(\frac{\Delta}{K}s,\gamma\right) \ \ \mbox{where} \ \ \gamma\equiv\frac{\Delta}{K},
\ \ \mbox{and} \ \ \sigma=\inf\{t|S(t)=K\}\wedge {T}.
\end{equation*}
From (\ref{2.22}) we obtain that if $s\geq K$
then the optimal perfect hedge is given by
$(\pi,\sigma)=((s+\Delta-K,0),0)$.
\end{exm}
\begin{figure}
\includegraphics[bb = 41 235 540 635, height = 2.5in]{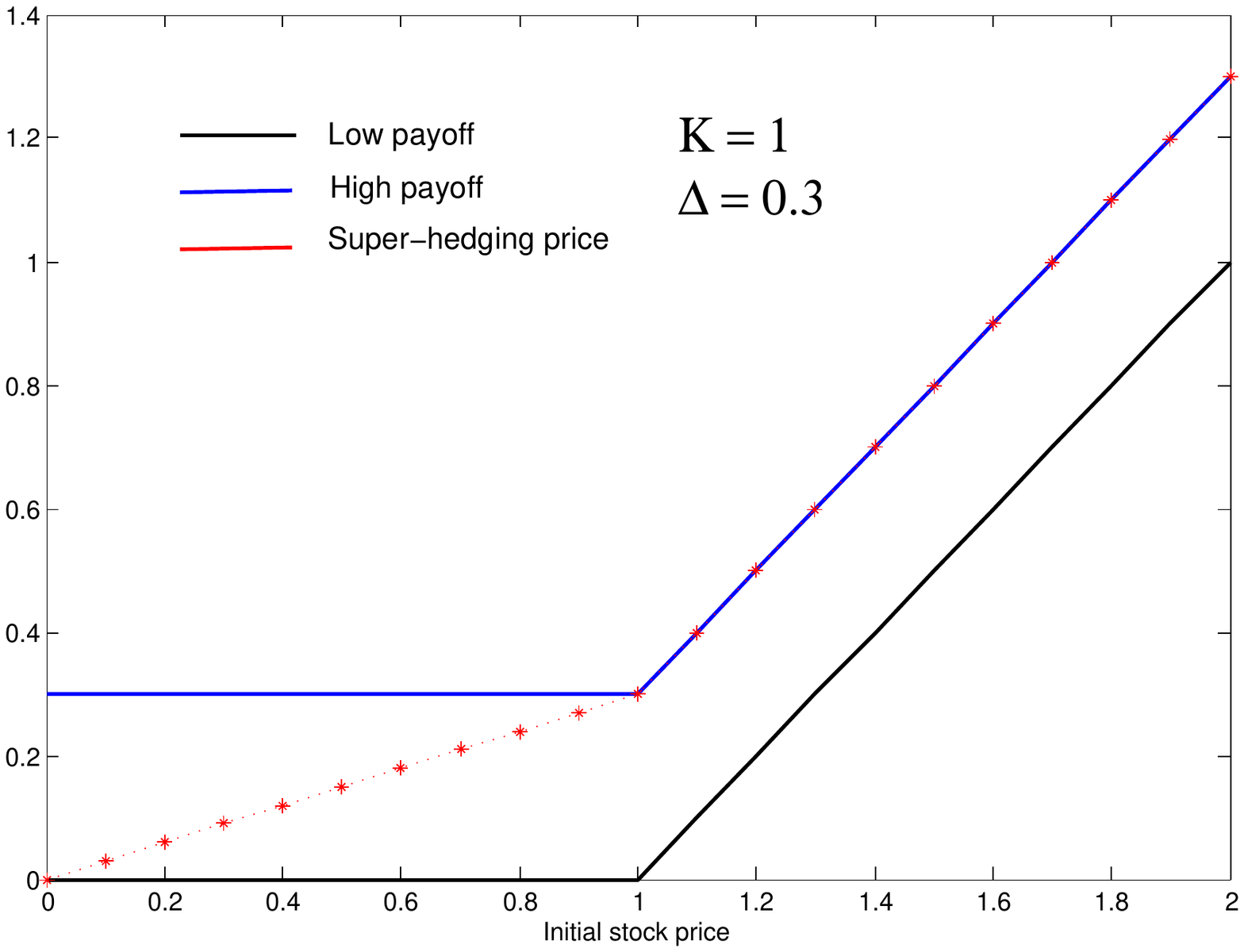}
\caption[]{Call option}
\end{figure}
\begin{exm}[Put option]\label{exm2.2}
Let $K>0$ be a constant and $d=1$.
Consider a game put option
with the discounted payoffs
\begin{equation*}\label{2.33}
Y(t)=\frac{(K-S(t))^{+}}{S_0(t)} \ \ \mbox{and} \ \ X(t)=Y(t)+\frac{\Delta}{S_0(t)}, \ \ t\in [0,T].
\end{equation*}
We consider two different cases.\\
i. $\Delta> K$. In this case we have $A=\infty$
and $B=0$. Thus
$R(x)\equiv K$ and the cheapest perfect hedge is given by
$\pi=(K,0)$ and $\sigma=T$.\\
ii. $\Delta\leq K$. In this case we have
$A=K$ and $B=-\frac{K-\Delta}{K}$.
This together with (\ref{2.20+}) yields (see Fig 2),
\begin{equation*}\label{2.34}
R(x)=\left(K-\frac{K-\Delta}{K}x\right)\mathbb{I}_{x<K}+\Delta\mathbb{I}_{x\geq K}.
\end{equation*}
Let $s\in\mathbb{R}_{+}$ be an initial position of the stock.
From (\ref{2.21}) we obtain that if $s<K$
then the optimal perfect hedge is given by
\begin{equation*}\label{2.35}
\pi=\left(K-\frac{K-\Delta}{K}s,\gamma\right) \ \ \mbox{where} \ \ \gamma\equiv-\frac{K-\Delta}{K} ,\ \
\mbox{and} \ \ \sigma=\inf\{t|S(t)=K\}\wedge {T}.\nonumber
\end{equation*}
From (\ref{2.22}) we obtain that if $s \geq K$
then the optimal perfect hedge is given by
$(\pi,\sigma)=((s+\Delta-K,0),0)$.
\end{exm}

\begin{figure}
\includegraphics[bb = 41 235 540 635, height = 2.5in]{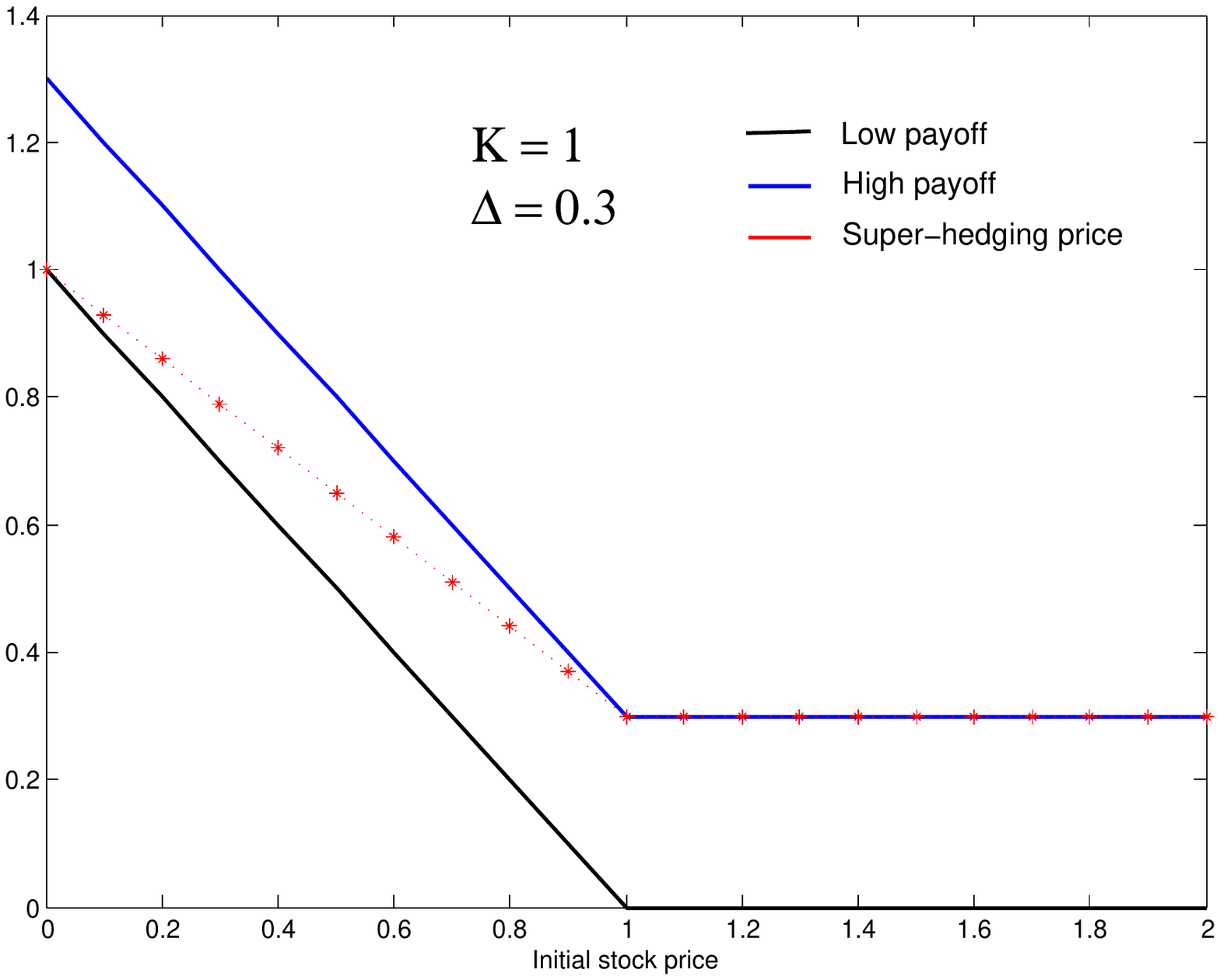}
\caption[]{Put option}
\end{figure}

Let us notice that in the above two examples, when the penalty $\Delta\geq K$,
the investor does not use his right to cancel. Namely in this case
the game option is essential a European/American option and the super--replication price
is given in terms of the standard concave envelope. For the case where $\Delta<K$
the super--replication price for game options is cheaper then in the European/American case
and we arrive at the game variant of the concave envelope.

\begin{exm}[Call--Put options]\label{exm2.2+}
Let $K>0$ be a constant and $d=2$.
Consider a game option with the discounted payoffs
\begin{equation*}\label{2.25}
Y(t)=\frac{(S_1(t)-S_2(t)+K)^{+}}{S_0(t)} \ \ \mbox{and} \ \ X(t)=Y(t)+\frac{\Delta}{S_0(t)}, \ \ t\in [0,T].
\end{equation*}
Namely, $F(x)=(x_1-x_2+K)^{+}$. Then
$F_1(y)=y+K$ and $F_2(y)=(K-y)^{+}$, $y\in\mathbb{R}_{+}$.
Again, we split the analysis into two different cases.\\
i. $\Delta> K$. In this case we have
$A_1=A_2=\infty$, $B_1=1$ and $B_2=0$.
From (\ref{2.20+}) we get
\begin{equation*}\label{2.27}
R(x)=(x_1+K)\mathbb{I}_{x\notin [\Delta-K,\infty)\times [\Delta,\infty) }+\big((x_1-x_2+K)^{+}+\Delta\big)
\mathbb{I}_{x\in [\Delta-K,\infty)\times [\Delta,\infty) }.
\end{equation*}
Let $s\in\mathbb{R}^2_{+}$ be an initial position of the stock.
If $s\notin [\Delta-K,\infty)\times [\Delta,\infty)$ then by (\ref{2.21}) we get
that the cheapest perfect hedge is
\begin{eqnarray*}\label{2.28}
&\pi=(s_1+K,\gamma) \ \ \mbox{where} \ \ \gamma\equiv\big(1,0) ,\\
&\mbox{and} \ \ \sigma=\inf\{t|S(t)\in [\Delta-K,\infty)\times [\Delta,\infty)  \}\wedge {T}.\nonumber
\end{eqnarray*}
If $s\in [\Delta-K,\infty)\times [\Delta,\infty)$ then by (\ref{2.22}) we get
that the optimal perfect hedge is given by
$(\pi,\sigma)=\big(\big(\big(s_1-s_2+K\big)^{+}+\Delta,0\big),0\big)$.\\
ii. $\Delta\leq K$. In this case we have
$A_1=\infty$, $A_2=K$, $B_1=1$ and $B_2=-\frac{K-\Delta}{K}$.
Thus
\begin{equation*}\label{2.29}
R(x)=\bigg(x_1+K-\frac{K-\Delta}{K}x_2\bigg)\mathbb{I}_{x_2<K }+\big(\big(x_1+K-x_2\big)^{+}+\Delta\big)
\mathbb{I}_{x_2\geq K}.
\end{equation*}
The optimal perfect hedge $(\pi,\sigma)$ for an initial position $s\in\mathbb{R}^2_{+}$ is given by
\begin{eqnarray*}\label{2.30}
&\pi=\big(s_1+K-\frac{K-\Delta}{K}s_2,\big(1,-\frac{K-\Delta}{K}\big)\big) \ \ \mbox{if} \ \ s_2<K,\\
&\pi=\big(\big(s_1-s_2+K\big)^{+}+\Delta,0\big) \ \ \mbox{if} \ \ s_2\geq K,\nonumber
\end{eqnarray*}
and $\sigma=\inf\{t|S_2(t)\geq K\}\wedge{T}$.
\end{exm}

Next, we show that for a non--bounded interest rate
process, our results which are given by Theorem \ref{thm2.1}
should not necessarily hold.
\begin{exm}\label{exm2.4}
Assume that our market consists of a bond given by the formula
$S_0(t)=\exp\left(\int_{0}^t|W(u)|du\right)$ and of one stock given by the formula
$S(t)=
\exp\left(\int_{0}^t(|W(u)|+2W(u))du\right)$ where
${\{W(t)\}}_{t=0}^T$ is a standard one--dimensional Brownian motion. Consider a put option
with the discounted payoffs $Y(t)=\frac{(0.5-S(t))^{+}}{S_0(t)}$ and
$X(t)=Y(t)+\frac{1}{S_0(t)}$. From Lemma 4.5 in \cite{GRS} it follows that
the discounted stock price $\frac{S(t)}{S_0(t)}= \exp\left(2\int_{0}^tW(u)du\right)$
satisfies Assumption \ref{asm2.1}. Observe that
$S_0(t)\geq \frac{1}{S(t)}$, and so if $Y(t)>0$ then $S_0(t)\geq 2$. Thus
$Y(t)\leq 1/4$, $t\in [0,T]$. We conclude that the
the super--replication price is not bigger than $1/4$ and so Theorem \ref{thm2.1} does not hold true (compare with Example
\ref{exm2.2}).
\end{exm}

The following example illustrates that for non--constant penalty game options,
Theorem \ref{thm2.1} does not necessarily hold true.
\begin{exm}\label{exm2.3}
Consider a game option with the discounted payoffs
\begin{equation*}\label{2.50}
Y(t)=\frac{1+(S(t)-3)^{+}}{S_0(t)} \ \ \mbox{and} \ \ X(t)=2Y(t), \ \ t\in [0,T].
\end{equation*}
Assume that the initial stock position is $s=4$ and the interest rate
process is a constant $r\geq 0$. We want to calculate the static super--hedging
price $\hat{V}:=\hat{V}(4)$.
Let $(\pi,\sigma)$ be a perfect and a trivial hedge.
If $\sigma\equiv T$ then the cheapest trivial hedge is achieved by
holding one stock at the initial moment of time,
and so the required initial capital is equal $4$.
If $\sigma\equiv 0$ then also the required capital is $4$.
Suppose that we want to find a static
perfect hedge with an initial capital less than $4$.
Clearly, the (constant) number of stocks in the corresponding portfolio
should satisfy $\gamma\geq 1$.
And so, there is no sense for the investor to cancel the contract when the stock
price is bigger than $4$
or smaller than $3$. Thus, without loss of generality
we assume that the investor cancelation time is of the form
$\sigma=\inf\{t|S(t)=\Lambda\}\wedge {T}$ where $3\leq\Lambda<4$.
The discounted stock price satisfies the conditional full support property,
and so we conclude that the super--replication property is given by
\begin{eqnarray*}
&(\hat{V}-4\gamma)\exp(rt)+\gamma s\geq s-2 \  \ \mbox{and}\\
& (\hat{V}-4\gamma)\exp(rt)+\gamma \Lambda\geq 2(\Lambda-2), \ s>\Lambda , \ t\leq T.
\end{eqnarray*}
Since $\gamma\geq 1$ and $\hat{V}- 4\gamma<0$ then the above relations are equivalent to the inequality
$(\hat{V}-4\gamma)\exp(rT)+\gamma \Lambda\geq 2(\Lambda-2)$. Finally, from the inequality
$4\exp(rT)>\Lambda$ we get that
the minimal value of $\hat{V}$ is attained by taking $\gamma=1$ and $\Lambda=3$. Thus
the cheapest cost of a trivial perfect hedge is given by
$\hat{V}=4-\exp(-rT)$. We conclude that the static super--replication price
depends on the interest rate
and so Theorem \ref{thm2.1} does not hold true.
\end{exm}

\section{Consistent price systems}\label{sec3}\setcounter{equation}{0}
It is well known that consistent price systems play a key role
in hedging with transaction costs.
We start with the definition.
\begin{dfn}\label{dfn3.0}
Let $\epsilon>0$. An absolutely continuous $\epsilon$--consistent price
system is a pair $(\hat{S},Q)$ which consists
of a probability measure $Q\ll P$ (absolutely continuous with respect to $P$)
and a $Q$--martingale $\hat{S}=\{\hat{S}_1(t),...,$
$\hat{S}_d(t)\}_{t=0}^T$
with values in $\mathbb{R}^d_{++}$
which satisfies
\begin{equation}\label{3.20}
1-\epsilon<\frac{\tilde S_i(t)}{\hat S_i(t)}<1+\epsilon, \ \ Q \ \mbox{a.s.}  \ \ i=1,...,d
\end{equation}
where recall that $\tilde{S}$ is the discounted stock process.
\end{dfn}
In this section we construct a general family of consistent price systems.
Before we state the main result of this section we need some preparations.
For a subset $D\subset\mathbb{R}^d$ we denote by
$conv(D)$ and $int(D)$, the convex hull of $D$ and the interior of $D$,
respectively.
Next, let $N\in\mathbb{N}$, $\epsilon>0$ and $M={\{(M_1(k),...,M_d(k))\}}_{k=0}^N$ be a finite valued martingale
with values in $\mathbb{R}^d_{++}$. We assume the following conditions.\\
i. $M(0)=\tilde{S}(0)$.\\
ii. For any $k<N$ the conditional distribution of
$M(k+1)-M(k)$ given $M(0),...,M(k)$ is an atomic distribution which contains
exactly $d+1$ points and satisfies
\begin{equation*}\label{3.20+-}
0 \in \ int \ conv \ supp \ \mathbb{P}(M(k+1)-M(k)|M(0),...,M(k)) \ \ \mbox{a.s.}
\end{equation*}
iii. For any $i=1,...,d$ and $k<N$
\begin{equation*}\label{3.20+-+}
(1-\epsilon/3)M_i(k)<M_i(k+1)<(1+\epsilon/3)M_i(k) \ \ \mbox{a.s.}
\end{equation*}
Observe that the martingale $M$ is defined on arbitrary probability space, not necessary
the same probability space on which $S$ is defined.

The following lemma is the main result of this section.
\begin{lem}\label{lem3.0}
There exists an absolutely continuous $\epsilon$--consistent price system
$(\hat{S},Q)$ such that the distribution of
the $Q$--martingale ${\{\hat{S}(kT/N)\}}_{k=0}^N$
equals to the distribution of ${\{M(k)\}}_{k=0}^N$.
Furthermore, for any $k<N$
\begin{eqnarray}\label{3.20+-++}
&Q\left(\hat{S}((k+1)T/N)|\mathcal{F}_{kT/N}\right)=\\
&Q\left(\hat{S}((k+1)T/N)|\hat{S}(0),\hat{S}(T/N),...,\hat{S}(kT/N)\right) \  \ {Q} \ \ \mbox{a.s.}\nonumber
\end{eqnarray}
where recall, ${\{\mathcal{F}_t\}}_{t=0}^T$ is the given filtration.
\end{lem}
\begin{proof}
Fix $\delta>0$. We will assume that $\delta$ is sufficiently small such that
the (Euclidean) distance between
any two different values of the random variables $M(0),...,M(N)$
is bigger than $2\delta(N+1)$. We denote the Euclidean norm by $||\cdot||$.
For $x,y\in\mathbb{R}^d_{++}$
and $k<N$ define the event
\begin{eqnarray*}\label{3.20+}
&A_{x,y,k}=\{||\tilde{S}(t)-(k+1-Nt/T)x-(Nt/T-k)y||<(k+1)\delta, \\
& t\in [kT/N,(k+1)T/N]\}.
\end{eqnarray*}
Denote by $\Psi(k,z_1,...,z_k)\subset\mathbb{R}^{d}_{++}$ the
(finite) set of all possible values of the random variable
$M(k+1)-M(k)$ given that $M(i)=z_i$, $i=1,...,k$.
Define on the probability space $(\Omega,\mathcal{F},P)$
the stochastic process ${\{\tilde{M}(k)\}}_{k=0}^N$
and the events $\mathcal{C}_0,...,\mathcal{C}_N$ by the following recursive relations,
$\tilde{M}(0)=\tilde{S}(0)$, $\mathcal{C}_0=\emptyset$, and for $k<N$
\begin{eqnarray}\label{3.21}
&\tilde{M}(k+1)=\tilde{M}(k)+\sum_{v\in\Psi(k,\tilde{M}(1),...,\tilde{M}(k))}(1-\mathbb{I}_{\mathcal{C}_k})
\mathbb{I}_{\{A_{\tilde{M}(k),v,k}\}}v\\
&\mbox{and} \ \ \mathcal{C}(k+1)=\{\tilde{M}(k+1)=\tilde{M}(k)\}.\nonumber
\end{eqnarray}
Observe that the sets $A_{\tilde{M}(k),v,k}$, $v\in\Psi(k,\tilde{M}(1),...,\tilde{M}(k))$
are disjoint.
Clearly $\mathcal{C}_0\subset\mathcal{C}_1\subset...\subset\mathcal{C}_N$.
 Fix $k\leq N$.
Let $\mathcal{C}\in \mathcal{F}_{kT/N}$ be an event such that
$\mathcal{C}\subset \Omega\setminus \mathcal{C}_k$ and the random variables
$\tilde{M}(0),...,\tilde{M}(k)$ are constants on $\mathcal{C}$.
From the definitions it follows that on the event $\Omega\setminus \mathcal{C}_k$,
we have $||\tilde{M}(k)-\tilde{S}(kT/N)||\leq k \delta$.
This together
with the conditional full support property of $\tilde{S}$
yields that on the event $\mathcal{C}$,
for any $v\in \Psi(k,\tilde{M}(1),...,\tilde{M}(k))$
\begin{eqnarray}\label{3.22}
&P(\tilde{M}(k+1)-\tilde{M}(k)=v|\mathcal{F}_{kT/N})\geq\\
&P\big(||\tilde{S}(t)-(k+1-Nt/T)\tilde{S}(kT/N)-\nonumber\\
&(Nt/T-k)v||<\delta|\mathcal{F}_{kT/N}\big)>0  \ \ \mbox{a.s.}\nonumber
\end{eqnarray}
From the second condition on the martingale $M$ we obtain (by induction)
that $0\in \  int \ conv \ \Psi(k,\tilde{M}(1),...,\tilde{M}(k))$ a.s. on the event
$\Omega\setminus \mathcal{C}_k$. Thus we conclude
that
\begin{equation}\label{3.23}
0\in \ int \ conv \ supp \ P(\tilde M(k+1)-\tilde{M}(k)|\mathcal{F}_{kT/N}) \ \mbox{for} \ \mbox{almost} \ \mbox{all} \
\omega\in \Omega\setminus \mathcal{C}_k.
\end{equation}
By using the Esscher transform in the same way that it was used in
Lemma 3.1 of \cite{GRS} we get that there exists a $\mathcal{F}_{kT/N}$ measurable
random vector $\theta(k)$ such that
\begin{equation}\label{3.23+}
E_P\left(\exp(\langle\theta(k),\tilde{M}(k+1)-\tilde{M}(k)\rangle)(\tilde{M}(k+1)-\tilde{M}(k))|\mathcal{F}_{kT/N}\right)=0
\end{equation}
where $E_P$ denotes the expectation with respect to $P$.
Set
\begin{equation}\label{3.23++}
Z(k)=E_P\left(\exp(\langle\theta(k),\tilde{M}(k+1)-\tilde{M}(k)\rangle)\mathbb{I}_{\Omega\setminus\mathcal{C}_{k+1}}|\mathcal{F}_{kT/N}\right).
\end{equation}
From (\ref{3.22}) it follows that on the event $\Omega\setminus \mathcal{C}_k$, $Z(k)>0$ a.s.
Define the stochastic process
\begin{equation*}\label{3.23+++}
\mathcal{H}(k)=\mathbb{I}_{\Omega\setminus \mathcal{C}_k}\prod_{i=0}^{k-1} \frac{\exp(\langle\theta(i),\tilde{M}(i+1)-\tilde{M}(i)\rangle)}
{Z(i)}, \ \ k=1,...,N.
\end{equation*}
Observe that
${\{\mathcal{H}(k)\}}_{k=1}^N$ is a martingale with $E_P \mathcal{H}(N)=1$. Thus there exists a
probability measure $Q\ll P$ which satisfies
\begin{equation*}\label{3.23++++}
\frac{dQ}{dP}|_{\mathcal{F}_{kT/N}}=\mathcal{H}(k).
\end{equation*}
From (\ref{3.23+}) it follows that ${\{\tilde{M}(k)\}}_{k=0}^N$
is a $Q$--martingale with respect to the
filtration ${\{\mathcal{F}_{kT/N}\}}_{k=0}^N$.
Define a $Q$--martingale ${\{\hat{S}(t)\}}_{t=0}^T$ by
\begin{equation*}\label{3.24-}
\hat{S}(t)=:E_Q(\tilde{M}(N)|\mathcal{F}_t)=E_Q(\tilde{M}(k+1)|\mathcal{F}_t), \ \ t\in [kT/N, (k+1)T/N], \ \ k<N
\end{equation*}
where $E_Q$ denotes the expectation with respect to the probability measure $Q$.
From the third condition on the martingale $M$ it follows
 that for any $k<N$  and $i=1,...,d$
\begin{equation}\label{3.24}
(1-\epsilon/3)\tilde{M}_i(k)<\tilde{M}_i(k+1)<(1+\epsilon/3)\tilde{M}_i(k) \ \ Q \ \ \mbox{a.s.}
\end{equation}
Let $\Theta\subset\mathbb{R}^d_{++}$ be the (finite) set of all possible values
of the random variables
$M(0),M(1)$\\
$,...,M(N)$. Since $\Theta$ is finite we can choose $\delta>0$
such that for any $0<\lambda<1$ and $x,y\in \Theta$
which satisfy $\max\left(\frac{x_i}{y_i},\frac{y_i}{x_i}\right)<\frac{1}{1-\epsilon/3}$, $i=1,...,d$,
we have the relation
\begin{eqnarray*}
&\left\{z\in\mathbb{R}^d:||z-\lambda x -(1-\lambda)y||<\delta(N+1)\right\}\subset\\
&\left\{z\in\mathbb{R}^d:\max\left(\frac{x_i}{z_i},\frac{z_i}{x_i}\right)<1+\epsilon/2, \ i=1,...,d\right\}.
\end{eqnarray*}
This together with (\ref{3.21}) and (\ref{3.24}) yields that for any $t\in [kT/N,(k+1)T/N]$ and $i=1,...,d$
\begin{equation}\label{3.24+}
(1-\epsilon/2)\tilde{M}_i(k)<\tilde{S}_i(t)<(1+\epsilon/2)\tilde{M}_i(k) \ \ Q \ \ \mbox{a.s.}
\end{equation}
From (\ref{3.24-})--(\ref{3.24}) we obtain
\begin{equation}\label{3.24++}
(1-\epsilon/3)\tilde{M}_i(k)<\hat{S}_i(t)<(1+\epsilon/3)\tilde{M}_i(k) \ \ Q \ \ \mbox{a.s.}
\end{equation}
By combining (\ref{3.24+})--(\ref{3.24++}) we arrive at (\ref{3.20}).
Finally, from the second assumption on $M$ we
obtain that the distribution of
${\{\tilde{M}(k)\}}_{k=0}^N$ (under $Q$)
equals to the distribution of ${\{M(k)\}}_{k=0}^N$ and
(\ref{3.20+-++}) holds true.
\end{proof}

\section{Robust optimal stopping and related convex analysis}\label{sec2+}\setcounter{equation}{0}
Let
$(\Omega_W, \mathcal{F}^{W}, {P}^{W})$
be a complete probability space together with a standard
$d$--dimensional Brownian motion
$W={\{(W_1(t),...,W_d(t))\}}_{t=0}^T$ and the right continuous filtration
$\mathcal{F}^{W}_t=\sigma\big\{\sigma{\{W(s)|s\leq{t}\}
\bigcup\mathcal{N}}\big\}$,
where $\mathcal{N}$ is the collection of all ${P}^W$ null sets.
For any $u\in [0,T]$ let $\mathcal{T}^W_{[0,u]}$
be the set of all stopping times with respect
to the Brownian filtration ${\{\mathcal{F}^W_t\}}_{t=0}^T$
which take values in the interval $[0,u]$. For any $x\in\mathbb{R}^d_{+}$ denote by $\Gamma(x)$ the set of
all bounded $d$--dimensional Brownian martingales ${\{M(t)\}}_{t=0}^T$,
such that for any $t$, $M(t)$
takes values in $\mathbb{R}^{d}_{+}$
and satisfies $M(0)=x$.
Define
\begin{equation}\label{2+.0}
\mathbb{V}(x):=\sup_{M\in\Gamma(x)}\inf_{\tau\in\mathcal{T}^W_{[0,T]}}E^W\left(F(M(\tau))+\Delta\mathbb{I}_{\tau<T}\right), \ \ x\in\mathbb{R}^d_{+}
\end{equation}
where $E^W$ denotes the expectation with respect to the probability measure $P^W$.
We will prove that $\mathbb{V}(s)=R(s)=\hat{V}(s)=V_{\kappa}(s)$. In this section we show the inequality
$\mathbb{V}(s)\geq R(s)\geq \hat{V}(s)$. First we prove that
the right hand side of (\ref{2+.0})
does not depend on the maturity date $T$.
\begin{lem}\label{lem2+.0}
For any $u\in (0,T]$,
\begin{equation}\label{2+.1}
\mathbb{V}(x)=\sup_{M\in\Gamma(x)}\inf_{\tau\in\mathcal{T}^W_{[0,u]}}E^W\left(F(M(\tau))+\Delta\mathbb{I}_{\tau<u}\right), \ \ x\in\mathbb{R}^d_{+}.
\end{equation}
\end{lem}
\begin{proof}
Let $x\in\mathbb{R}^d_{+}$ and $u\in (0,T]$. Set $\alpha=\frac{T}{u}$.
Consider the Brownian motion given by
$\tilde{W}(t):=\frac{1}{\sqrt\alpha} W(\alpha t)$, $t\in [0,u]$.
Let $\{{\mathcal{F}}^{\tilde{W}}_t\}_{t=0}^{u}$ be the (usual) filtration which is generated by $\tilde{W}$
and let ${\mathcal T}^{\tilde W}_{[0,u]}$ be the
a set of all stopping times with values in $[0,u]$ with respect to this filtration. For any
$x\in\mathbb{R}^d_{+}$ denote by $\tilde{\Gamma}(x)$ the set of all martingales ${\{M(t)\}}_{t=0}^{u}$
with respect to $\{\mathcal{F}^{\tilde W}_t\}_{t=0}^{u}$
such that for any $t$, $M(t)$
takes values in $\mathbb{R}^{d}_{+}$
and satisfies $M(0)=x$.
Observe that for any $t\in [0,u]$, ${\mathcal{F}}^{\tilde W}_t=\mathcal{F}^W_{\alpha t}$.
Define the maps
$\Psi:\Gamma(x)\rightarrow \tilde\Gamma(x)$ and
$\Phi:\mathcal{T}^W_{[0,T]}\rightarrow {\mathcal T}^{\tilde W}_{[0,u]}$ by
\begin{equation*}\label{2+.2}
\Psi(M)(t)=M(\alpha t), \ t\in [0,u] \ \ \mbox{and} \ \ \Phi(\tau)=\frac{\tau}{\alpha}.
\end{equation*}
Observe that $\Psi$ and $\Phi$ are bijections, and $\Psi(M)(\Phi(\tau))=M(\tau)$ for any
$M\in\Gamma(x)$ and $\tau\in\mathcal{T}^W_{[0,T]}$.
Thus
\begin{eqnarray*}\label{2+.3}
&\mathbb{V}(x)=\sup_{M\in\Gamma(x)}\inf_{\tau\in\mathcal{T}^W_{[0,T]}}
E^W\left(F(M(\tau))+\Delta\mathbb{I}_{\tau<T}\right)=\\
&\sup_{M\in\Gamma(x)}\inf_{\tau\in\mathcal{T}^W_{[0,T]}}
E^W\left(F\left(\Psi(M)(\Phi(\tau))\right)+\Delta\mathbb{I}_{\Phi(\tau)<u}\right)=\nonumber\\
&\sup_{M\in\tilde\Gamma(x)}\inf_{\tau\in\mathcal{T}^{\tilde W}_{[0,u]}}
E^W\left(F(M(\tau))+\Delta\mathbb{I}_{\tau<u}\right)=\nonumber\\
&\sup_{M\in\Gamma(x)}\inf_{\tau\in\mathcal{T}^{W}_{[0,u]}}E^W\left(F(M(\tau))+\Delta\mathbb{I}_{\tau<u}\right).\nonumber
\end{eqnarray*}
\end{proof}
In the next lemma we prove several properties of the function $\mathbb{V}$ which will be essential in the Proof of Theorem \ref{thm2.1}.
\begin{lem}\label{lem2+.2}
The function $\mathbb{V}$ is satisfying $\mathbb{V}\in\mathcal{G}$,
where recall that $\mathcal{G}$ was defined after
(\ref{2.15+}).
\end{lem}
\begin{proof}
Consider the stopping time $\tilde\tau=0$ and the martingale $\tilde{M}\equiv x$. Clearly,
\begin{eqnarray*}\label{2+.7+}
&F(x)=\inf_{\tau\in\mathcal{T}^W_{[0,T]}}E^W\left(F(\tilde{M}(\tau))+\Delta\mathbb{I}_{\tau<T}\right)\leq \mathbb{V}(x)\leq\\
&\sup_{M\in\Gamma(x)}E^W\left(F({M}(\tilde\tau))+\Delta\mathbb{I}_{\tilde\tau<T}\right)=F(x)+\Delta.\nonumber
\end{eqnarray*}
Next, for any $x\in\mathbb{R}^d_{+}$
define the bijection $\Upsilon_x:\Gamma(1,...,1)\rightarrow\Gamma(x)$
by $\Upsilon_x(M)=(x_1 M_1,...,x_d M_d)$.
Since $F$ is a Lipschitz continuous function, there is a constant $\tilde{L}$
such that for any $x,y\in\mathbb{R}^d_{+}$
\begin{eqnarray*}\label{2+.9}
&|\mathbb{V}(y)-\mathbb{V}(x)|\leq\\
&\sup_{M\in\Gamma(1,...,1)}\sup_{\tau\in\mathcal{T}^W_{[0,T]}}
E^W\big{|}F \left(\Upsilon_y(M)(\tau)\right)-
F \left(\Upsilon_x(M)(\tau)\right)\big{|}\leq\\
&\sup_{M\in\Gamma(1,...,1)}\sup_{\tau\in\mathcal{T}^W_{[0,T]}}\sum_{i=1}^d \tilde{L}|y_i-x_i|E^WM_i(\tau)=
\tilde{L}\sum_{i=1}^d |y_i-x_i|.\nonumber
\end{eqnarray*}
Thus $\mathbb{V}$ is continuous and satisfying $F\leq \mathbb{V}\leq F+\Delta$.
Finally, we prove that if $\mathbb{V}<F+\delta$ in a convex region $D$ then
$\mathbb{V}$ is concave in $D$.
Let $x^{(1)},x^{(2)},x^{(3)}\in D$ such that
$x^{(3)}=\lambda x^{(1)}+(1-\lambda)x^{(2)}$ for some $0<\lambda<1$.
Choose $\epsilon>0$. From Lemma \ref{lem2+.0} it follows that there exist martingales
$M_i\in\Gamma(x^{(i)})$, $i=1,2$ such that
\begin{equation}\label{2+.10}
\mathbb{V}(x^{(i)})<\epsilon+\inf_{\tau\in\mathcal{T}^W_{[0,T/2]}}E^W
\left(F(M(\tau))+\Delta\mathbb{I}_{\tau<T/2}\right) , \ \ i=1,2.
\end{equation}
For any $i=1,2$ let
$\phi_i:{C}([0,T/2];\mathbb{R}^d)\rightarrow {C}([0,T/2];\mathbb{R}^d)$
be a map such that $\phi_i(t,y)$ depends only on the restriction of $y$ to the interval $[0,t]$,
and
\begin{equation*}\label{2+.11}
M^{(i)}_{|[0,T/2]}=\phi_i(W_{|[0,T/2]}),  \ i=1,2.
\end{equation*}
Consider the Brownian motion $W^{(1)}(t)=W(t+T/2)-W(T/2)$, $t\in [0,T/2]$.
Let $A\subset \mathbb{R}^d$ be such that $P^W(W(T/2)\in A)=\lambda$.
Define the function $f:[0,T/2]\times\mathbb{R}^d\rightarrow\mathbb{R}^d$ by
$f(t,y)=x^{(1)}P^W(y+W(T/2-t)\in A)+x^{(2)}P^W(y+W(T/2-t)\notin A)$. Observe that
the stochastic process ${\{M(t)\}}_{t=0}^T$ defined by
\begin{eqnarray*}\label{2+.11}
&M(t):=f(t,W(t))
 \ \mbox{for} \ t\in [0,T/2)\\
&\mbox{and} \ \ M_{|[T/2,T]}:=\mathbb{I}_{W(T/2)\in A}\phi_1({W}^{(1)})+ \mathbb{I}_{W(T/2)\notin A}\phi_2({W}^{(1)})\nonumber
\end{eqnarray*}
is a martingale which satisfies $M\in \Gamma(x^{(3)})$.
Next, define the stochastic processes
\begin{equation*}\label{2+.12}
U^{(i)}(t)=ess \inf_{\tau\in\mathcal{T}^W_{[0,T/2]},\tau\geq t}E^W(F(M^{(i)}(\tau))+
\Delta\mathbb{I}_{\tau<T/2}|\mathcal{F}^W_t), \ t\in [0,T/2], \ i=1,2
\end{equation*}
and
\begin{equation*}\label{2+.13}
U(t)=ess \inf_{\tau\in\mathcal{T}^W_{[0,T]},\tau\geq t}E^W(F(M(\tau))+
\Delta\mathbb{I}_{\tau<T}|\mathcal{F}^W_t), \ t\in [0,T].
\end{equation*}
Define the stopping time  $\tilde\tau\in\mathcal{T}^W_{[0,T]}$ by,
\begin{equation*}\label{2+.14}
\tilde\tau=\inf\{t|U(t)=F(M(t))+\Delta\}\wedge{T}.
\end{equation*}
From the general theory of optimal stopping (see \cite{PS}, Chapter I) it follows that
\begin{equation*}\label{2+.15}
U(0)=E^W(F(M(\tilde\tau))+
\Delta\mathbb{I}_{\tilde\tau<T}).
\end{equation*}
Fix $0<v<T/2$. Define the Brownian motions
${W}^{(2)}(t)=W(t+v)-W(v)$, $t\in [0,T/2-v]$
and $W^{(3)}(t)=W(t+T/2-v)-W(T/2-v)$, $t\in [0,T/2]$.
Observe that  for any $t<T/2-v$,
$E^W\left(f(T/2,y+W(T/2-v))|\mathcal{F}^W_t\right)=
f(t+v,y+W(t))$, and so we can define the martingale
$M^{(y)}\in\Gamma(f(v,y))$
by
\begin{eqnarray*}\label{2+.15+-}
&M^{(y)}(t)=f(t+v,y+W(t)), \ \mbox{for} \ t\in [0,T/2-v),\\
&M^{(y)}_{|[T/2-v,T-v]}=\mathbb{I}_{y+{W}(T/2-v)\in A}\phi_1(W^{(3)})+
\mathbb{I}_{y+{W}(T/2-v)\notin A}\phi_2(W^{(3)}),\nonumber\\
&\mbox{and} \ \ M^{(y)}(t)=M^{(y)}(T-v) \ \ \mbox{for} \ t\in [T-v,T].\nonumber
\end{eqnarray*}
Clearly
\begin{eqnarray*}\label{2+.15+}
&M(t+v)=f(t+v,W(v)+{W}^{(2)}(t)), \ \mbox{for} \ t\in [0,T/2-v),\\
&\mbox{and} \ \ M_{|[T/2,T]}=\mathbb{I}_{W(v)+{W}^{(2)}(T/2-v)\in A}\phi_1({W}^{(1)})+\nonumber\\
&\mathbb{I}_{W(v)+{W}^{(2)}(T/2-v)\notin A}\phi_2({W}^{(1)}).\nonumber
\end{eqnarray*}
From the fact that
the Brownian motions, ${\{W(t)\}}_{t=0}^{T/2-v}$ and
${\{W^{(3)}(t)\}}_{t=0}^{T/2}$ are independent,
 ${\{W(t)\}}_{t=0}^{T/2}$ and ${\{W^{(1)}(t)\}}_{t=0}^{T/2}$
are independent, and ${\{W(t)\}}_{t=0}^{v}$ and ${\{W^{(2)}(t)\}}_{t=0}^{T/2-v}$ are independent,
we obtain that
\begin{eqnarray}\label{2+.15++}
U(v)=\psi(W(v))
\end{eqnarray}
where
\begin{equation}\label{2+.16}
\psi(y):=\inf_{\tau\in\mathcal{T}^W_{[0,T-v]}}E^W\left(F(M^{(y)}(\tau))+\Delta\mathbb{I}_{\tau<T-v}\right), \ y\in\mathbb{R}^d.
\end{equation}
Since $M(v)\in D$ for any $v\in [0,T/2]$, from
Lemma \ref{lem2+.0} and (\ref{2+.15++})--(\ref{2+.16}) we get that
$U(v)\leq \mathbb{V}\left(f(v,W(v))\right)=\mathbb{V}(M(v))< F(M(v))+\Delta$ for any $v\in [0,T/2]$.
Thus $\tilde\tau\geq T/2$. This together with (\ref{2+.10}) and the fact that $W^{(1)}$
is independent of $\mathcal{F}^W_{T/2}$ yields
\begin{eqnarray*}\label{2+.17}
&\mathbb{V}(x^{(3)})\geq U(0)=E^WU(T/2)=\lambda U^{(1)}(0)+(1-\lambda)U^{(2)}(0)\geq \\
&\lambda \mathbb{V}(x^{(1)})+(1-\lambda)\mathbb{V}(x^{(2)})-\epsilon,\nonumber
\end{eqnarray*}
and by taking $\epsilon\downarrow 0$ we complete the proof.
\end{proof}
Next, we provide some convex analysis for the set $\mathcal{G}$
and the static super--replication price $\hat{V}(s)$.
\begin{lem}\label{lem2+.3}
The function $R$ which is defined by (\ref{2.20+})
is the minimal element of $\mathcal{G}$.
\end{lem}
\begin{proof}
We will use induction on the dimension $d$. Let $d=1$.
From Lemma 2.4 in \cite{EV} it follows that $\mathcal{G}$
has a minimal element and from the fact that $F$
is convex it follows that the
minimal element is equal to $R$ which is given by (\ref{2.20+}).
Next, we prove that if the result is true
for any $d\leq n$, then it is true for $d=n+1$.
Assume by contradiction that the claim is false.
Thus there exists a function $g\in\mathcal{G}$
and $x\in\mathbb{R}^d_{++}$ such that $g({x})< R({x})$. Set,
$v=\inf\{t\geq 0|g(t{x})<R(t{x})\}$ and let ${y}=v{x}$. We argue that $||{y}||<H({y})$
(where $H$ was defined after (\ref{2.20+})). Indeed, if (by contradiction)
$||{y}||\geq H({y})$, then $H({y})<\infty$ and $g\big(H({y}){y}/||{y}||\big)\geq
R\big(H({y}){y}/||{y}||\big)=F\big(H({y}){y}/||{y}||\big)+\Delta$, thus
$g\big(H({y}){y}/||{y}||\big)=F\big(H({y}){y}/||{y}||\big)+\Delta$.
Define the function
$f(u)=F(u{y})+\Delta-g(u{y})$, $u\in [H({y})/||{y}||,\infty)$.
Since there exists some $\delta>0$ for which $f(1+\delta)>0$,
then from the fact that $F$ is convex and
$g\in\mathcal{G}$ we get that $f$ is a strictly
increasing convex function on the interval $[1+\delta,\infty)$, and so for sufficiently
large $u$ we will get that $f(u)>\Delta$, which is a
contradiction to the fact that $g\geq F$. Thus we conclude that
$H(y)>|| y||$, which means that there exist $\epsilon>0$ and
$\tilde{y}\in\mathbb{R}^d_{++}$ such that
\begin{equation}\label{2+.18+}
g(\tilde y)<F(0)+\langle \tilde y,B\rangle<F(\tilde y)+\Delta.
\end{equation}
Let $\xi\in\partial F(\tilde y)$ and consider the
hyperplane $K=\{\tilde{x}\in\mathbb{R}^d|\langle \tilde{x}-\tilde{y},\xi-B\rangle=0\}$,
(where $B$ is the vector which is given by (\ref{2.20})).
From (\ref{2+.18+}) and the convexity of $F$ it follows that
\begin{equation}\label{2+.18++}
F(\tilde x)-\langle \tilde x,B\rangle\geq F(\tilde{y})-\langle \tilde{y} ,B\rangle\geq F(0)-\Delta, \ \ \forall\tilde x\in K.
\end{equation}
Clearly, there is a point on $K$ of the form $z=(0,...,\alpha,0,...,0)$ for some $\alpha\geq 0$.
Consider the half--line $\mathcal{L}=\{z^{(\lambda)}:={z}+\lambda(\tilde{y}-{z})|\lambda\in\mathbb{R}_{+}\}\subset K$.
Define $\lambda_1=\inf\{\lambda\geq 0|z^{(\lambda)}\notin\mathbb{R}^d_{+}\}$
and $\lambda_2=\inf\{0\leq\lambda\leq \lambda_1|F\big(z^{(\lambda)}\big)+\Delta=
g\big(z^{(\lambda)}\big)\}$, where $\lambda_1$, $\lambda_2$
equal to $\infty$ if the corresponding sets are empty. We distinguish
between cases.\\
i. If $\lambda_1=\lambda_2=\infty$, then
$g<F+\Delta$
on the half--line $\{z^{(\lambda)}|\lambda\in\mathbb{R}_{+}\}$, and so
by applying the induction assumption for $d=1$ we obtain
\begin{equation}\label{2+.18+++}
g(z)\geq F(0)+\langle  z,B\rangle.
\end{equation}
Since the function $g$ is concave
on the half--line $\mathcal{L}$, then from (\ref{2+.18++})--(\ref{2+.18+++})
we get that
$g(\tilde y)\geq F(0)+\langle  \tilde y,B\rangle$, which is a contradiction to (\ref{2+.18+}).\\
ii. If $\lambda_1=\infty$ and $\lambda_2<\infty$, then by using the fact that $g(\tilde y)< F(\tilde y)+\Delta$
and a similar argument to the one before (\ref{2+.18+}) we obtain that
$\lambda_2>1$, in particular (\ref{2+.18+++}) is valid for this case as well. By
applying (\ref{2+.18++})--(\ref{2+.18+++}) together with the fact that $g$ is concave on the line segment
$\{z^{(\lambda)}|\lambda\in [0,\lambda_2]\}$ we get that
$g(\tilde y)\geq F(0)+\langle \tilde y,B\rangle$, which is a contradiction to (\ref{2+.18+}).\\
iii. Let $\lambda_1<\infty$ and $\lambda_2=\infty$. In this case (\ref{2+.18+++}) remains true. Since
$\tilde{y}\in\mathbb{R}_{++}$ then $\lambda_1>1$.
From the induction assumption we get that
$g(z^{(\lambda_1)})\geq R(z^{(\lambda_1)})=F(0)+\langle z^{(\lambda_1)}, B\rangle$. This together with
(\ref{2+.18+++}) and the fact
that $g$ is concave on the line segment
$\{z^{(\lambda)}|\lambda\in [0,\lambda_1]\}$ yields
$g(\tilde{y})\geq F(0)+\langle\tilde{y},B\rangle$,
which is a contradiction to (\ref{2+.18+}).\\
iv. Finally, let $\lambda_2\leq\lambda_1<\infty$.
From (\ref{2+.18++}) and the induction assumption it follows that
\begin{eqnarray}\label{2+.19}
&g(z^{(\lambda_1)})\geq F(0)+\langle z^{(\lambda_1)},B\rangle \ \ \mbox{and} \\
&g(z^{(\lambda_2)})= F(z^{(\lambda_2)})+\Delta \geq F(0)+\langle z^{(\lambda_2)},B\rangle.\nonumber
\end{eqnarray}
Define $\hat{\lambda}=\sup\{\lambda\leq\lambda_1|g(z^{(\lambda)})=F(z^{(\lambda)})+\Delta\}$.
Since $g$ is concave on the line segments
$\{z^{(\lambda)}|\lambda\in [0,\lambda_2]\}$
and $\{z^{(\lambda)}|\lambda\in [\hat\lambda,\lambda_1]\}$,
and $g=F+\Delta$ on the line segment
$\{z^{(\lambda)}|\lambda\in [\lambda_2,\hat\lambda]\}$, then from
(\ref{2+.18++}) and (\ref{2+.19}) we obtain that
$g(\cdot)\geq F(0)+\langle\cdot,B\rangle$ on the line segment
$\{z^{(\lambda)}|\lambda\in [0,\lambda_1]\}$, which is a contradiction to (\ref{2+.18+}).
\end{proof}
In the following lemma we show that there is
a trivial perfect hedge with an initial capital $R(s)$, where $s$ is the
initial stock position.
\begin{lem}\label{lem2+.4}
Let $s\in\mathbb{R}^d_{+}$ be an initial stock position. The
hedge $(\pi,\sigma)$ which is defined according to
(\ref{2.21})--(\ref{2.22}) is a perfect hedge.
\end{lem}
\begin{proof}
Without loss of
generality we assume that there exists $0 \leq j\leq d$ such
that $A_i<\infty$ if and only if $i\leq j$.
First we prove the following relations. For any $x\in\mathbb{R}^d_{+}$,
\begin{eqnarray}\label{2+.20}
&i. \ F(x)>F(0)+\langle x,B \rangle\Rightarrow
 \sum_{i=1}^j \frac{x_i}{A_i}>1. \\
&ii. \  \sum_{i=1}^j \frac{x_i}{A_i}=1\Rightarrow F(x)+\Delta\leq F(0)+\langle x,B \rangle.\nonumber
\end{eqnarray}
Indeed if $\sum_{i=1}^{j}\frac{x_i}{A_i}<1$ then
from the convexity of $F$ we obtain
\begin{eqnarray*}\label{2+.32+++}
&F(x)\leq \sum_{i=1}^j \frac{x_i}{A_i}F_i(A_i)+
\mathbb{I}_{j<d}\frac{1-\sum_{i=1}^j x_i/A_i}{d-j} \times\\
&\sum_{k=j+1}^d F_k\left(\frac{x_k(d-j)}{1-\sum_{i=1}^j x_i/A_i}\right)\leq
F(0)+\langle x,B \rangle.
\nonumber
\end{eqnarray*}
This proves (by contradiction) the first statement in (\ref{2+.20}).
Next, let $\sum_{i=1}^j \frac{x_i}{A_i}=1$.
Fix $0<\epsilon<1$. Consider the vector $y=(1-\epsilon) x$.
From the convexity of $F$ it follows that
\begin{eqnarray*}\label{2+.34}
&F(y)+\Delta\leq \sum_{i=1}^j \frac{y_i}{A_i}\big(F_i(A_i)+\Delta\big)+\mathbb{I}_{j<d}\frac{\epsilon}{d-j} \\
&\times\sum_{k=j+1}^d \big(F_k(y_k(d-j)/ \epsilon)+\Delta\big)\leq F(0)+\langle y, B \rangle+\epsilon\Delta\nonumber
\end{eqnarray*}
and by letting $\epsilon\downarrow 0$ we obtain the second statement in (\ref{2+.20}).
Now, we are ready to prove the lemma. Let $(\pi,\sigma)$ be the hedge which is given by
(\ref{2.21})--(\ref{2.22}).
If $R(s)=F(s)+\Delta$ then the statement is trivial. Assume that $R(s)<F(s)+\Delta$, then
from (\ref{2+.20}) we get
\begin{equation}\label{2+.35}
\sum_{i=1}^j \frac{s_i}{A_i}<1.
\end{equation}
Let $t\in [0,T]$. Clearly, on the event $\sigma<t$
we have
\begin{eqnarray*}
&V^\pi(\sigma)=F(0)+\langle \tilde{S}(\sigma),B\rangle\geq\frac{1}{S_0(\sigma)}
\left(F(0)+\langle S(\sigma),B\rangle\right)\\
&\geq \frac{1}{S_0(\sigma)}\left(\Delta+F(S(\sigma))\right)=H(\sigma,t).
\end{eqnarray*}
Consider the event $t\leq\sigma$. From (\ref{2+.35}) and the second statement
in (\ref{2+.20})
it follows that for any $v<\sigma$,
$\sum_{i=1}^j \frac{S_i(v)}{A_i}<1$.
Thus by applying the first statement in (\ref{2+.20})
we get
\begin{equation*}
V^\pi(t)\geq\frac{1}{S_0(t)}\big(F(0)+\langle S(t),B\rangle\big)\geq\frac{1}{S_0(t)} F(S(t))=H(\sigma,t).
\end{equation*}
Since $t$ was arbitrary the proof is completed.
\end{proof}
From Lemmas \ref{lem2+.2}--\ref{lem2+.4} we obtain the following
Corollary.
\begin{cor}\label{cor2+.1}
For any $s\in\mathbb{R}^d_{+}$,
\begin{equation}\label{2.37}
\mathbb{V}(s)\geq R(s)\geq\hat{V}(s).
\end{equation}
Furthermore, the hedge $(\pi,\sigma)$ which is defined according to
(\ref{2.21})--(\ref{2.22}) is a perfect hedge.
\end{cor}

\section{Optimal Stopping and Price Consistent Systems}\label{sec5}\setcounter{equation}{0}
Let $\mathbb{M}_d$ be the space
of $d\times d$ matrices with the operator norm
$||\mathcal{A}||=\sup_{||v||=1}||\mathcal{A}(v)||$. We denote by $\hat{I}$ the unit matrix.
For any $i\leq d$, let $e_i:=(0,...,0,1,0,...,0)$ be the
unit vector where $1$ is in the
$i$--th place.
For a matrix $\mathcal{A}\in\mathbb{M}_d$ and a vector $x\in\mathbb{R}^d$
we denote by $\mathcal{A}\cdot x$ the matrix multiplication between $\mathcal{A}$
and the column vector $x$.
First, we review some basic concepts from the weak convergence theory.
For any $c\grave{a}dl\grave{a}g$
stochastic process ${\{X(t)\}}_{t=0}^T$ with values in some Euclidean space
$\mathbb{R}^m$, denote by $\mathbb{P}^X$ the distribution of $X$ on the canonical
space $\mathbb{D}([0,T];\mathbb{R}^m)$ equipped with the Skorohod
topology (for details see \cite{B2}) i.e. for any Borel set $D\subset
\mathbb{D}([0,T];\mathbb{R}^m)$, $\mathbb{P}^X(D)=\mathbb{P}\{X\in{D}\}$.
The usual filtration which is generated by the process $X$ will be denoted by
${\{\mathcal{F}^X_t\}}_{t=0}^T$.
For
a sequence of ($\mathbb{R}^m$ valued) stochastic processes
$X^{(n)}$ we will
use the notation $X^{(n)}\Rightarrow{X}$ to indicate that the
probability measures $\mathbb{P}^{X^{(n)}}$, $n\geq{1}$ converge
weakly to $\mathbb{P}^X$
on the space $\mathbb{D}([0,T];\mathbb{R}^m)$. For convergence of optimal stopping values
we will need
a stronger form of convergence, than the standard weak convergence.
This form is called "extended weak convergence" and was introduced in \cite{A}.
In \cite{A} Aldous introduced the notion of "extended weak
convergence" via prediction processes. He showed that the original condition
is equivalent to a more elementary condition which does
not require the use of prediction processes (see \cite{A}, Proposition
16.15). We will use the latter condition as a definition.
\begin{dfn}\label{dfn5.1}
A sequence of $X^{(n)}$, $n\in\mathbb{N}$
extended weak converges to a stochastic process $X$ if for any
continuous bounded functions
$\psi_1,...,\psi_k\in{C(\mathbb{D}([0,T];\mathbb{R}^d))}$
\begin{equation*}
(X^{(n)},Z^{(n,1)},...,Z^{(n,k)})\Rightarrow (X,Z^{(1)},...,Z^{(k)})
\end{equation*}
on the space $\mathbb{D}([0,T];\mathbb{R}^{d+k})$,
where for any $t\leq{T}$, $1\leq i\leq{k}$ and $n\in\mathbb{N}$
\begin{equation*}
Z^{(n,i)}_t=\mathbb{E}_n(\psi_i(X^{(n)})|\mathcal{F}^{X^{(n)}}_t),
 \ \mbox{and} \ Z^{(i)}=\mathbb{E}(\psi_i(X)|\mathcal{F}^X_t)
\end{equation*}
$\mathbb{E}_n$ denotes the expectation
on the probability space on which $X^{(n)}$ is defined
and $\mathbb{E}$ denotes the expectation on the probability space on
which $X$ is defined. We will denote extended weak
convergence by $X^{(n)}\Rrightarrow X$.
\end{dfn}

Next, consider the Brownian
probability space $(\Omega_W,\mathcal{F}^W,P^W)$.
Let $\mathcal{L}$ be the set of all $\mathbb{M}_d$ valued
adapted processes (to the Brownian filtration)
$\alpha={\{\alpha_{ij}(t)\}}_{1\leq i,j\leq d, t\in [0,T]}$,
given by $\alpha(t)=f(t,W)$, $t\in [0,T]$
where $f=f:[0,T]\times C([0,T];\mathbb{R}^d)\rightarrow\mathbb{M}_d$
is a continuous bounded function and satisfies
$f(t,x)=f(t,y)$ if $x(u)=y(u)$ for any $u\in [0,t]$.
The above condition guarantees that $\alpha$ is an adapted
(to the Brownian filtration) process.
Observe that we consider $W$ as a random element
in $C([0,T];\mathbb{R}^d)$.
Finally denote by
 $\Gamma_b(x)\subset \Gamma(x)$ as a set of all Brownian martingales
 $M={\{M_1(t),...,M_d(t)\}}_{t=0}^T\in\Gamma(x)$ of the form
\begin{equation*}\label{5.2}
M_i(t)=x_i\exp\left(\int_{0}^t \sum_{j=1}^d \alpha_{ij}(u)dW_j(u)-\frac{1}{2}\int_{0}^t
\sum_{j=1}^d \alpha^2_{ij}(u)du\right).
\end{equation*}
Next, let $\mathbb{A}$ be a $(d+1)\times (d+1)$ orthogonal matrix such that its last
column equals to $(\frac{1}{\sqrt{d+1}},...,\frac{1}{\sqrt{d+1}})^{*}$.
Let $\Omega_{\xi}={\{1,2,...,d+1\}}^\infty$ be the space of infinite
sequences $\omega=(\omega_1,\omega_2,...)$;
$\omega_i\in{\{1,2,...,d+1\}}$ with the product probability
$P^{\xi}={\{\frac{1}{d+1},...,\frac{1}{d+1}\}}^\infty$. Define a
sequence of i.i.d. random vectors $\xi(1),\xi(2),...$ by
$\xi(i)(\omega)=\sqrt{d+1}(A_{\omega_i 1},A_{\omega_i
2}...,A_{\omega_i d})$, $i\in\mathbb{N}$. Denote by
$\mathcal{T}_n$ the set of all stopping times with respect to the filtration
generated by $\xi(i)$, $i\in\mathbb{N}$, with values in the set
$\{0,1,...,n\}$.
Notice that the random vectors $\xi(k)$, $k\in\mathbb{N}$
have mean zero and a covariance matrix which is equals to $\hat{I}$.
Choose $\alpha(\cdot)=f(\cdot,W)\in\mathcal{L}$.
Let $\lambda_n\downarrow 0$ be a sequence such that
for any $0\leq k\leq n$ the matrix
$\lambda_n\hat{I}+f\left(kT/n,\sqrt\frac{T}{n}\sum_{i=1}^k \xi(i)\right)$
is non--singular $P^{\xi}$ a.s.
 Clearly, there exists such sequence since
the set of all eigenvalues of matrices of the form
$f\left(kT/n,\sqrt\frac{T}{n}\sum_{i=1}^k \xi(i)\right)$
is countable.
For any $n\in\mathbb{N}$ define
the martingale ${\{M^{(n)}(k)\}}_{k=0}^n$
by $M^{(n)}(0)=s$, and for $k<n$
\begin{eqnarray*}\label{5.4}
&M^{(n)}_i(k+1)=M^{(n)}_i(k)\bigg(1+\sqrt\frac{T}{n}\bigg\langle\lambda_n e_i+\\
&f_i\left(kT/n,\sqrt\frac{T}{n}\sum_{i=1}^k \xi(i)\right),\xi(k+1)\bigg\rangle\bigg), \ \ i=1,...,d \nonumber
\end{eqnarray*}
where $f_i$ is the $i$--th row of the matrix $f$.
We assume that $n$ is sufficiently large such that $M^{(n)}$
takes on values in $\mathbb{R}^d_{++}$. Recall that $s=(s_1,...,s_d)$ is the initial stock position.
\begin{lem}\label{lem5.1}
Set, $W^{(n)}(t)=\sqrt\frac{T}{n}\sum_{i=1}^{[nt/T]} \xi(i)$, $t\in [0,T]$.
We have
\begin{equation}\label{5.6}
\left(W^{(n)}(t),M^{(n)}([nt/T])\right)_{t=0}^T \Rightarrow \left(W(t),M(t)\right)_{t=0}^T
\end{equation}
on the space $\mathbb{D}([0,T];\mathbb{R}^d)\times\mathbb{D}([0,T];\mathbb{R}^d)$ (with the product topology).
\end{lem}
\begin{proof}
Define the ($\mathbb{R}^d$ valued) processes
$Y(t):=\int_{0}^t f(u,W(u-))\cdot dW(u)$ and
$Y^{(n)}(t)=\int_{0}^t \left(\lambda_n \hat{I}+f(u,W^{(n)}(u-))\right) \cdot dW^{(n)}(u)$,
$t\in [0,T]$, $n\in\mathbb{N}$.
From \cite{H} it follows that $W^{(n)}\Rightarrow W$ on the space
$\mathbb{D}([0,T];\mathbb{R}^d)$. This together with Theorem 4.1
in \cite{DP} yields that $(W^{(n)},Y^{(n)})\Rightarrow (W,Y)$
on the space
$\mathbb{D}([0,T];\mathbb{R}^d)\times\mathbb{D}([0,T];\mathbb{R}^d)$.
Next, observe that the stochastic process
$\hat{M}^{(n)}(t):=M^{(n)}([nt/T])$, $t\in [0,T]$ is the unique solution of the SDE,
$\hat{M}^{(n)}_i(t)=s_i+\int_{0}^t \hat{M}^{(n)}_i(u-)dY^{(n)}_i(u)$, $i\leq d$
and ${\{M(t)\}}_{t=0}^T$ is the unique solution of the SDE,
$M_i(t)=s_i+\int_{0}^t M_i(u)dY_i(u)$, $i\leq d$.
Thus by applying Theorem 4.4 in \cite{DP} we obtain (\ref{5.6}).
\end{proof}
Next, we use the extended weak convergence theory in order to
treat optimal stopping values.
\begin{lem}\label{lem5.2}
For any $\delta>0$ there exists an absolutely continuous $\delta$--consistent price
system $(\hat{S},Q)$ which satisfies
\begin{equation}\label{5.7}
\inf_{\sigma\in\mathcal{T}_{[0,T]}}E_Q\left(F(\hat{S}(\sigma))+\Delta \mathbb{I}_{\sigma<T}\right)\geq
\inf_{\sigma\in\mathcal{T}^W_{[0,T]}}E^W\left(F(M(\sigma))+\Delta \mathbb{I}_{\sigma<T}\right)-\delta.
\end{equation}
\end{lem}
\begin{proof}
Let $\delta>0$. The processes $W$ and $W^{(n)}$, $n\in\mathbb{N}$
are processes with independent increments and so, from Proposition 20.18 in \cite{A}
the usual weak convergence
$W^{(n)}\Rightarrow W$ implies extended weak convergence
$W^{(n)}\Rrightarrow W$. Since for any $n$
the process ${\{M^{(n)}([nt/T])\}}_{t=0}^T$ is adapted to the filtration generated by
$W^{(n)}$ we get from Lemma \ref{lem5.1} that
$\left(W^{(n)}(t),M^{(n)}([nt/T])\right)_{t=0}^T$
$\Rrightarrow (W(t),M(t))_{t=0}^T$.
Now, that we established extended weak convergence, we apply
Theorem 17.2 in \cite{A} and obtain
\begin{eqnarray}\label{5.8}
&\lim_{n\rightarrow\infty}\min_{\sigma\in\mathcal{T}_n} E^\xi \left(F(M^{(n)}(\sigma))+\Delta \mathbb{I}_{\sigma<n}\right)\\
&=\inf_{\sigma\in\mathcal{T}^{W}_{[0,T]}} E^W\left(F(M(\sigma))+\Delta \mathbb{I}_{\sigma<T}\right)\nonumber
\end{eqnarray}
where $E^\xi$ is the expectation with respect to $P^\xi$.
Assume that $F$ is Lipschizt
continuous with a constant $\tilde{L}$, namely
$|F(y)-F(z)|\leq \tilde{L}\sum_{i=1}^d |y_i-z_i|.$
Observe that for sufficiently large $n$ the martingale $M^{(n)}$ satisfies
the three conditions before
Lemma \ref{lem3.0}, where for the third condition we take
$\epsilon:=\frac{\delta}{1+\tilde{L} \sum_{i=1}^d s_i}$. Thus from (\ref{5.8}) it follows that
we can choose $N$ which satisfies the above and the inequality
\begin{equation}\label{5.8+}
\min_{\sigma\in\mathcal{T}_N} E^\xi \left(F(M^{(N)}(\sigma))+\Delta \mathbb{I}_{\sigma<N}\right)>
\inf_{\sigma\in\mathcal{T}_{[0,T]}} E^W\left(F(M(\sigma))+\Delta \mathbb{I}_{\sigma<T}\right)-\delta/2.
\end{equation}
From Lemma \ref{lem3.0}
we obtain that there exists an absolutely continuous $\epsilon$--price
consistent
system $(Q,\hat{S})$ which satisfies this lemma
for the martingale $M^{(N)}$.
Denote by $\mathcal{T}^N\subset \mathcal{T}_{[0,T]}$
the set of stopping times with values in the set
$\{0,T/N,2T/N,...,T\}$.
Observe that the fact that $M^{(N)}$ satisfies the second
condition before Lemma 3.2 implies that the filtration which is generated by
$M^{(N)}$ coincides with the filtration generated by $W^{(N)}$. Thus
from the standard dynamical programming for
optimal stopping (see \cite{PS} Chapter I)
and the equality (\ref{3.20+-++}) we obtain
\begin{equation}\label{5.9}
\inf_{\sigma\in\mathcal{T}^N}
E_Q\left(F(\hat{S}(\sigma))+\Delta \mathbb{I}_{\sigma<T}\right)=
\min_{\sigma\in\mathcal{T}_N} E^\xi \left(F(M^{(N)}(\sigma))+\Delta \mathbb{I}_{\sigma<n}\right).
\end{equation}
Next, for any stopping time $\sigma\in\mathcal{T}_{[0,T]}$
define the stopping time
$\phi_N(\sigma)=\min\{k|kT/N\geq\sigma\}\frac{T}{N}\in\mathcal{T}^N$.
Similarly to (\ref{3.24++})
we obtain $|\hat{S}_i(\sigma)-\hat{S}_i(\phi_n(\sigma))|$
$\leq \epsilon \hat{S}^{(n)}_i(\sigma)/3$, $i=1,...,d$.
This together with (\ref{5.8+})--(\ref{5.9}) yields
\begin{eqnarray*}\label{5.10}
&\inf_{\sigma\in\mathcal{T}_{[0,T]}}E_Q\left(F(\hat{S}(\sigma))+\Delta \mathbb{I}_{\sigma<T}\right)\geq\\
&\inf_{\sigma\in\mathcal{T}_{[0,T]}}E_Q\left(F(\hat{S}(\phi_n(\sigma)))+\Delta \mathbb{I}_{\phi_n(\sigma)<T}\right)
-\epsilon \tilde{L} \sum_{i=1}^d s_i/3=\\
&\inf_{\sigma\in\mathcal{T}^N}E_Q\left(F(\hat{S}(\sigma))+\Delta \mathbb{I}_{\sigma<T}\right)
-\epsilon \tilde{L} \sum_{i=1}^d s_i/3\geq\\
&\inf_{\sigma\in\mathcal{T}^W_{[0,T]}}E^W\left(F(M(\sigma))+\Delta \mathbb{I}_{\sigma<T}\right)-\delta.
\end{eqnarray*}
\end{proof}
By using standard density arguments
it follows that $\Gamma_b(x)$ is dense in $\Gamma(x)$.
Namely, for any $M\in\Gamma(x)$ there exists a sub--sequence
${\{M^{(n)}\}}_{n=1}^\infty\subset\Gamma_b(x)$ such that
$$\lim_{n\rightarrow\infty}E^W\left(\sup_{0\leq t\leq T} ||M^{(n)}(t)-M(t)||\right)=0.$$
Thus from Lemma \ref{2+.0} we obtain that for any $u<T$
$$\mathbb{V}(x)=\sup_{M\in\Gamma_b(x)}\inf_{\sigma\in\mathcal{T}^W_{[0,u]}}
 E^W \left(F(M(\sigma))+\Delta\mathbb{I}_{\sigma<u}\right), \ \ x\in\mathbb{R}^d_{++}.$$
Next, we notice that if, in
formula (\ref{5.7}) we put some $u\in [0,T]$ instead of $T$ then Lemma \ref{lem5.2}
still remains true (and can be proved in a similar way).
In view of Lemma 4.1, we arrive at the following Corollary.
\begin{cor}\label{cor5.1}
For any $\epsilon>0$ and $u<T$ there exists an absolutely continuous $\epsilon$--price consistent
system $(Q,\hat{S})$ which satisfies
\begin{equation}\label{5.10}
\inf_{\sigma\in\mathcal{T}_{[0,u]}}E_Q\left(F(\hat{S}(\sigma))+\Delta \mathbb{I}_{\sigma<u}\right)\geq \mathbb{V}(s)-\epsilon
\end{equation}
where $\mathcal{T}_{[0,u]}$ is the set of all stopping times with respect
to the filtration ${\{\mathcal{F}_t\}}_{t=0}^T$ with values in the interval
$[0,u]$.
\end{cor}

\section{Proof of Main  results}\label{sec4}\setcounter{equation}{0}
In this Section we complete the proof of
Theorem \ref{thm2.1}.
In view of Corollary \ref{cor2+.1} it remains to show that
$V_{\kappa}(s)\geq\mathbb{V}(s)$.
Let $(\pi=(\Xi,\gamma),\sigma)\in\mathcal{A}(\Xi)\times\mathcal{T}_{[0,T]}$
be a perfect hedge. We want to show that
$\Xi\geq\mathbb{V}(s)$, where $s$ is the initial stock position.
Let $\epsilon>0$ be such that $\frac{1+\epsilon}{1-\epsilon}<1+\kappa$.
Since the interest rate process ${\{r(t)\}}_{t=0}^T$
is bounded there exists $\hat T\in [0,T]$ such that
\begin{equation}\label{4.1}
\exp\left(\int_{0}^{\hat T} r(u)du\right)<1+\epsilon \ \ P \ \ \mbox{a.s.}
\end{equation}
From Corollary \ref{cor5.1} we obtain that there exists an
absolutely continuous $\epsilon$ price consistent
system $(Q,\hat{S})$ which satisfies (\ref{5.10}).
From (\ref{3.20}) we get
\begin{equation}\label{4.9+}
E_Q(\tilde{S}_i(\tau_2)|\mathcal{F}_{\tau_1})\geq
\frac{1}{1+\epsilon}E_Q(\hat{S}_i(\tau_2)|\mathcal{F}_{\tau_1})=\frac{1}{1+\epsilon}\hat{S}_i(\tau_1)\geq
(1-\kappa)\tilde{S}_i(\tau_1).
\end{equation}
Similarly
\begin{equation}\label{4.9++}
E_Q(\tilde{S}_i(\tau_2)|\mathcal{F}_{\tau_1})\leq
(1+\epsilon)E_Q(\hat{S}_i(\tau_2)|\mathcal{F}_{\tau_1})=(1+\epsilon)\hat{S}_i(\tau_1)\leq
(1+\kappa)\tilde{S}_i(\tau_1).
\end{equation}
Next, for any $k\in\mathbb{N}$ define the stopping time
\begin{equation}\label{4.10}
\tau_k=\sigma\wedge \hat{T}\wedge\inf\left\{t|\sum_{i=1}^d \tilde{S}_i(t)+\sum_{i=1}^d\int_{[0,t]} |d\gamma^{(n)}_i|\geq k\right\}.
\end{equation}
For any $m\in\mathbb{N}$ consider the partition $b_{m,l}=lT/m$, $l=0,1....,m$.
From (\ref{4.9+}) and the dominated convergence theorem we obtain
\begin{eqnarray}\label{4.12+}
&E_Q\int_{[0,\tau_{k}]}\left\langle \tilde{S}(u),d\gamma_{-}(u)\right\rangle=\\
&\lim_{m\rightarrow\infty}E_Q\left(\sum_{l=0}^{m-1} \left\langle \tilde{S}(\tau_{k}\wedge b_{m,l+1}),
\left(\gamma(\tau_{k}\wedge b_{m,l+1})-\gamma(\tau_{k}\wedge b_{m,l})\right)\right\rangle\right)\nonumber\\
&\leq\frac{1}{1-\kappa}\lim_{m\rightarrow\infty}
E_Q\left(\sum_{l=0}^{m-1}\left\langle\tilde{S}_{\tau_k}, \left(\gamma(\tau_{k}\wedge b_{m,l+1})-
\gamma(\tau_{k}\wedge b_{m,l})\right)\right\rangle\right)\nonumber\\
&=\frac{1}{1-\kappa}E_Q\left(\left\langle \tilde{S}(\tau_{k}),\int_{[0,\tau_{k}]}d\gamma_{-}(u)\right\rangle\right), \ \ k\in\mathbb{N}.\nonumber
\end{eqnarray}
In a similar way we obtain
\begin{equation}\label{4.12++}
E_Q\int_{[0,\tau_{k}]}\left\langle\tilde{S}(u),d\gamma_{+}(u)\right\rangle\geq\frac{1}{1+\kappa}
E_Q\left(\left\langle\tilde{S}(\tau_{k}),\int_{[0,\tau_{k}]}d\gamma_{+}(u)\right\rangle\right).
\end{equation}
From (\ref{4.12+})--(\ref{4.12++}) it follows that
\begin{eqnarray}\label{4.13}
&E_Q\bigg(\Xi+\left\langle \gamma(\tau_{k}), \tilde{S}(\tau_{k})\right\rangle-\left\langle \gamma(0),s\right\rangle-(1+\kappa)\\
&\times\int_{[0,\tau_{k}]}\left\langle \tilde{S}(u),d\gamma_{+}(u)\right\rangle+(1-\kappa)
\int_{[0,\tau_{k}]}\left\langle \tilde{S}(u),d\gamma_{-}(u)\right\rangle\bigg)\leq\nonumber\\
&\Xi+E_Q\left\langle \gamma(0), \tilde{S}(\tau_{k})-s\right\rangle\leq \Xi+\epsilon \sum_{i=1}^d s_i|\gamma_i(0)|\nonumber
\end{eqnarray}
Since $(\pi,\sigma)$ is a perfect hedge, the term in the brackets in formula (\ref{4.13}) is non--negative.
Thus, from (\ref{4.13}), the Fatou's lemma and the fact that $\tau_k\uparrow \sigma\wedge \hat{T}$ as $k\rightarrow\infty$,
we get
\begin{eqnarray}\label{4.15}
&\Xi+\epsilon \sum_{i=1}^d s_i|\gamma_i(0)|\geq\\
&E_Q\bigg(\Xi+\left\langle \gamma(\sigma\wedge \hat{T}), \tilde{S}(\sigma\wedge \hat{T})\right\rangle-\left\langle \gamma(0),s\right\rangle-(1+\kappa)\nonumber\\
&\times\int_{[0,\sigma\wedge \hat{T}]}\left\langle \tilde{S}(u),d\gamma_{+}(u)\right\rangle+(1-\kappa)
\int_{[0,\sigma\wedge \hat{T}]}\left\langle \tilde{S}(u),d\gamma_{-}(u)\right\rangle\bigg)\geq\nonumber\\
&E_Q\left(H(\sigma,\hat{T})\right),\nonumber
\end{eqnarray}
where the last inequality follows from the definition of a perfect hedge.
From (\ref{3.20}), (\ref{4.1}) and the convexity of $F$ we get
\begin{eqnarray*}
&E_Q \left(H(\sigma,\hat{T})\right)\geq \frac{1}{1+\epsilon}
E_Q \left(F(\tilde{S}(\sigma\wedge\hat{T}))+\Delta\mathbb{I}_{\sigma<\hat{T}}-\epsilon F(0)\right)\geq\\
&\frac{1}{1+\epsilon}
E_Q \left(F(\hat{S}(\sigma\wedge\hat{T}))+\Delta\mathbb{I}_{\sigma<\hat{T}}-\epsilon \left(F(0)+\tilde{L}\sum_{i=1}^d s_i\right )\right).
\end{eqnarray*}
Notice that in the last inequality we used the Lipschitz continuity
of $F$ and the fact that $\hat{S}$ is a $Q$ martingale.
This together with (\ref{5.10}) and (\ref{4.15}) yields
$$\Xi\geq \frac{1}{1+\epsilon}
\left(\mathbb{V}(s)-\epsilon \left(1+F(0)+\tilde{L}\sum_{i=1}^d s_i\right)\right)-\epsilon\sum_{i=1}^d s_i|\gamma_i(0)|$$
and by letting $\epsilon\downarrow 0$
we obtain $\Xi\geq\mathbb{V}(s)$,
as required.\qed

\section*{Acknowledgements}
I would like to thank Paolo Guasoni
for insightful discussions.

\end{document}